\documentclass[aps,twocolumn,superscriptaddress,footinbib]{revtex4}
\usepackage{epstopdf}
\usepackage{graphicx}  
\usepackage{bm}        
\usepackage{braket}
\usepackage{exscale}                  
\usepackage[intlimits]{amsmath}       
\usepackage{amsfonts}
\usepackage{amssymb,amscd}
\usepackage{amsthm}
\usepackage{color}
\usepackage{hyperref}
\usepackage{subfigure}
\usepackage[normalem]{ulem}
\usepackage{slashed}
\usepackage{leftidx}
\usepackage{xcolor,soul}
\usepackage{txfonts}
\usepackage{charter}
\usepackage{dsfont}
\usepackage{mathrsfs}
\newtheorem{theorem}{Theorem}
\newtheorem{lemma}{Lemma}
\newtheorem{corollary}{Corollary}

\newtheorem{proposition}{Proposition}
\newtheorem{definition}{Definition}
\newtheorem{remark}{Remark}

\usepackage{times}

\hypersetup{
    colorlinks=true,
    linkcolor=blue,
    citecolor=blue,
    urlcolor=blue
}
\begin{document}
\title{Unitary-accessible coherence-erasure distance: Exact qubit solution and a tight qutrit bound}
\author{Dong-Ping Xuan\footnote{E-mail: xuandongping@ccut.edu.cn}}
\affiliation{School of Mathematics and Statistics, Changchun University of Technology, Changchun 130012, Jilin, China}
\author{Hua Nan\footnote{E-mail: nanhua@ybu.edu.cn}}
\affiliation{Department of Mathematics, College of Sciences, Yanbian University, Yanji 133002, China}
\author{Zhi-Xi Wang\footnote{E-mail: wangzhx@cnu.edu.cn}}
\affiliation{School of Mathematical Sciences, Capital Normal University, Beijing 100048, China}
\author{Shao-Ming Fei\footnote{E-mail: feishm@cnu.edu.cn}}
\affiliation{School of Mathematical Sciences, Capital Normal University, Beijing 100048, China}
\affiliation{Max-Planck-Institute for Mathematics in the Sciences, 04103 Leipzig, Germany}
\begin{abstract}
\textbf{
We introduce the unitary-accessible coherence-erasure distance, which
measures the minimum unitary-control cost required to transform a
quantum state into an incoherent state without changing its spectrum.
The target states are restricted to incoherent states on the same
unitary orbit, leading to a geometry different from conventional
state-space proximity.
We establish inequalities relating orbit-restricted quantum transport,
Bures geometry, and the unitary-control distance. These quantities
coincide for pure states but can differ for mixed states.
For qubits, we derive an exact expression and show that, in the nondegenerate
case, the coherence-erasure cost depends on the eigenbasis orientation
rather than on the eigenvalues.
For nondegenerate qutrits, we reduce the problem to an optimization
over the monomial unitary group. We derive the tight universal bound
$\arccos[(2\sqrt2-1)/4]$ and construct an explicit eigenbasis that
attains it. The resulting distance also determines the minimum
coherence-erasure time under bounded unitary driving, giving it a
direct operational meaning.
}
\end{abstract}
\maketitle

\vspace{.4cm}


\section{Introduction}
\label{sec:introduction}

Quantum coherence is a fundamental feature of quantum mechanics and
has been identified as an important resource for quantum information
processing~\cite{SASAKI7501,ABERG150402,BCP,BAGAN160406,YuanHouTang2020,LiSun2018}.
A systematic framework for the quantification of quantum coherence
was introduced by Baumgratz \emph{et al.}~\cite{BCP}. In this
framework, incoherent states are defined with respect to a fixed
reference basis, incoherent operations are identified as the
corresponding free operations, and a set of conditions is proposed
for a proper coherence measure. This resource-theoretic framework has
since been further developed and reviewed in a variety of
settings~\cite{Streltsov2017}.

Within this framework, a standard distance-based approach quantifies
coherence by minimizing the distance from a given state to the set of
incoherent states. Given a suitable distance measure $D$, one defines
 $C_D(\rho) = \min_{\delta\in\mathcal I} D(\rho,\delta)$,
where $\mathcal I$ denotes the set of incoherent states~\cite{BCP,GuoGuanLiHuFei2026}.
This construction characterizes the proximity of $\rho$ to the
incoherent set in state space.

Such a state-space minimization does not take into account whether the
target state is accessible under the allowed dynamics. This matters
when the evolution is restricted to unitary transformations. Since
unitary evolution preserves the spectrum, a state $\rho$ cannot in
general be transformed into an arbitrary incoherent state. If
$\mathcal O_\rho$ denotes the unitary orbit of $\rho$, the accessible
incoherent states are restricted to
$\mathcal I_\rho=\mathcal I\cap\mathcal O_\rho$. Therefore, the distance
to the full incoherent set does not in general reflect the unitary
control required to reach an incoherent state.

Related problems have been studied from several perspectives. Modi and
Gu considered unitary orbits in the study of quantum correlations and
used them to distinguish coherent and incoherent contributions
\cite{ModiGu2013}. Distances between isospectral mixed states have also
been studied in connection with quantum evolution and quantum speed
limits~\cite{AnderssonHeydari2013,Marvian2016,Hornedal2022}.
A closely related quantity is the quantum charging distance introduced
in Ref.~\cite{Gyhm2024}. It gives the minimum unitary control required
to transform one given state into another state with the same spectrum
when the driving strength is bounded. For pure states, this distance
coincides with the Bures angle, while for mixed states the two
distances can be different.


This leads to a natural question in the study of quantum coherence:
what is the minimum unitary-control cost required to erase coherence
while preserving the spectrum? Existing distance-based coherence
measures minimize over the full incoherent set and therefore do not
account for unitary accessibility. Conversely, the quantum charging
distance quantifies the control cost between two prescribed isospectral
states, but does not optimize over incoherent targets.

Motivated by this question, we introduce the unitary-accessible coherence-erasure distance
 $D_{\mathcal I}^{U}(\rho) = \min_{\delta\in\mathcal I_\rho} D_U(\rho,\delta)$,
where $D_U$ is the unitary-control distance between states with the
same spectrum. The quantity $D_{\mathcal I}^{U}(\rho)$ measures the
minimum unitary cost required to transform $\rho$ into an incoherent
state while preserving its spectrum. It therefore combines the
fixed-basis description of coherence with the control cost of
transformations between states on the same unitary orbit.

The term ``coherence erasure'' has also been used for processes in
which coherence is removed by noise or decohering operations
\cite{SinghBeraMisraPati2015}. The situation considered here is
different. We remove coherence relative to a fixed reference basis
using only closed system unitary evolution. As a result, the spectrum,
purity, and entropy of the state remain unchanged during the process.
Accordingly, $D_{\mathcal I}^{U}$ is used here to characterize the
unitary cost of coherence erasure rather than as a coherence measure
in the usual resource-theoretic sense.

The relation between coherence, quantum resources, and evolution time
has also been studied in the context of quantum speed limits
\cite{Marvian2016,Campaioli2022,MaiYu2024}. In contrast, our problem
requires the final state to be incoherent while keeping the same
spectrum as the initial state.

We obtain analytical results for both qubits and qutrits. For
nondegenerate qubits, the distance is obtained exactly and depends
only on the orientation of the eigenbasis relative to the fixed
reference basis. For nondegenerate qutrits, the problem reduces to an
optimization over monomial unitaries and yields a tight universal
bound. We further compare the distance with orbit-restricted
state-space distances and connect it with the minimum
coherence-erasure time under bounded unitary driving.

The paper is organized as follows.
Sec.~\ref{sec:preliminaries} introduces the notation and the geometric
quantities used throughout the paper.
Sec.~\ref{sec:framework} defines the unitary-accessible
coherence-erasure distance and establishes its general properties and
its relation to orbit-restricted distances.
Sec.~\ref{sec:qubit} presents the exact analysis for qubits.
Sec.~\ref{sec:qutrit} treats the qutrit problem and Sec.~\ref{sec:operational_examples} determines its
extremal configuration.
Finally, Sec.~\ref{sec:conclusion} summarizes the results.

\section{Preliminaries}

\label{sec:preliminaries}
We first recall the basic notions used throughout this work~\cite{Nielsenbook}, including
incoherent states and unitary orbits, the Bures angle, quantum
charging distance, quantum optimal transport cost, and the matrix
notation needed in the subsequent analysis.

\paragraph*{Incoherent states and unitary orbits.}
Let $\mathcal H$ be a $d$-dimensional Hilbert space equipped with a
fixed reference basis $\{\ket{i}\}_{i=1}^d$~\cite{XiHuLiFan2018}. We denote by
$\mathcal D(\mathcal H)$ the set of density operators on
$\mathcal H$, and by $U(d)$ the group of $d\times d$ unitary
matrices. The set of incoherent states~\cite{BCP,SharmaPati2018} is
\begin{equation}
    \mathcal I
    =
    \left\{
    \delta=\sum_{i=1}^{d}p_i\ket{i}\bra{i}:
    p_i\ge0,\quad \sum_{i=1}^{d}p_i=1
    \right\}.
    \label{eq:incoherent_set}
\end{equation}
For a density operator $\rho$, its unitary orbit is
\begin{equation}
    \mathcal O_\rho
    =
    \left\{
    U\rho U^\dagger:\;U\in U(d)
    \right\}.
    \label{eq:unitary_orbit}
\end{equation}
Two density operators belong to the same unitary orbit if and only
if they have the same spectrum, including multiplicities~\cite{ModiGu2013,PhysRevA.78.042108}.

\paragraph*{Fidelity and the Bures angle.}
We adopt the squared Uhlmann fidelity~\cite{Uhlmann1976,Jozsa1994},
\begin{equation}
    F(\rho,\sigma)
    =
    \left[
    \operatorname{Tr}
    \sqrt{\sqrt{\rho}\sigma\sqrt{\rho}}
    \right]^2.
    \label{eq:fidelity}
\end{equation}
With this convention, the Bures angle~\cite{Nielsenbook} is defined as
\begin{equation}
    \theta_B(\rho,\sigma)
    =
    \arccos\sqrt{F(\rho,\sigma)}.
    \label{eq:bures_angle}
\end{equation}
For pure states $\rho=|\psi\rangle\langle\psi|$ and
$\sigma=|\phi\rangle\langle\phi|$, this reduces to
$\theta_B(\rho,\sigma)=\arccos|\langle\psi|\phi\rangle|$.

\paragraph*{Quantum charging distance.}
For two isospectral density operators $\rho$ and $\sigma$, the
quantum charging distance introduced in Ref.~\cite{Gyhm2024} is
defined as the minimum time required to transform $\rho$ into
$\sigma$ by unitary evolution under a normalized driving Hamiltonian $V_t$,
\begin{equation}
    D_{\rm CD}(\rho,\sigma)
    =
    \min_{\substack{V_t:\,\|V_t\|_\infty=1}}
    T,
    \label{eq:charging_distance_time}
\end{equation}
where $\sigma =U_T\rho U_T^\dagger$, $U_T =\mathcal T\exp\left( -i\int_0^T V_t\,dt\right)$,
and we set $\hbar=1$. Here, $\|\cdot\|_\infty$ denotes the
operator norm.

An equivalent expression is
\begin{equation}
    D_{\rm CD}(\rho,\sigma)
    =
    \min_{U:\,U\rho U^\dagger=\sigma}
    \left\|
        i\operatorname{ln}U
    \right\|_\infty.
    \label{eq:charging_distance}
\end{equation}
The minimum can always be achieved by a time-independent driving Hamiltonian.

If $\rho$ and $\sigma$ are nondegenerate and have spectral
decompositions $\rho = \sum_{j=1}^{d} \lambda_j\ket{r_j}\bra{r_j}$,  $\sigma = \sum_{j=1}^{d}\lambda_j\ket{s_j}\bra{s_j}$,
then the optimization reduces to
\begin{equation}
    D_{\rm CD}(\rho,\sigma)
    =
    \min_{\phi_1,\ldots,\phi_d}
    \left\|
        i\operatorname{ln}
        \left[
            \sum_{j=1}^{d}
            e^{i\phi_j}
            \ket{s_j}\bra{r_j}
        \right]
    \right\|_\infty.
    \label{eq:CD_nondegenerate}
\end{equation}
Therefore, as long as the eigenspaces and their multiplicities remain
unchanged, the charging distance depends only on the eigenspaces and
not on the corresponding eigenvalues.

For notational simplicity, for pure states we write
$D_{\rm CD}(\psi,\varphi)$ for
$D_{\rm CD}(\ket{\psi}\bra{\psi},\ket{\varphi}\bra{\varphi})$.
In this case,
\begin{equation}
    D_{\rm CD}(\psi,\varphi)
    =
    \arccos|\langle\psi|\varphi\rangle|
    =
    \theta_B(\psi,\varphi).
    \label{eq:CD_pure}
\end{equation}
More generally, for any two isospectral states,
\begin{equation}
    \theta_B(\rho,\sigma)
    \le
    D_{\rm CD}(\rho,\sigma),
    \label{eq:Bures_CD_bound}
\end{equation}
with equality for pure states.
\paragraph*{Quantum optimal transport cost.}
We briefly review the quantum optimal transport cost that will be
used later to compare different notions of distance in coherence
erasure~\cite{ZhangLi2025}.

Let $\rho_A$ and $\sigma_B$ be density operators on $d$-dimensional
Hilbert spaces $\mathcal H_A$ and $\mathcal H_B$, respectively.
A bipartite state
$\chi_{AB}\in\mathcal D(\mathcal H_A\otimes\mathcal H_B)$ is called
a \emph{quantum coupling} of $(\rho_A,\sigma_B)$ if
$\operatorname{Tr}_A\chi_{AB}=\sigma_B,
    \quad
    \operatorname{Tr}_B\chi_{AB}=\rho_A$.
We denote the set of all quantum couplings of
$(\rho_A,\sigma_B)$ by $\mathcal B(\rho_A,\sigma_B)$. This set is
always nonempty, since it contains, for example,
$\rho_A\otimes\sigma_B$.

On $\mathcal H_A\otimes\mathcal H_B$, the SWAP operator is defined as
\begin{equation}
    \mathsf S
    =
    \sum_{i,j=1}^{d}
    \ket{ij}\bra{ji}.
    \label{eq:swap_operator}
\end{equation}
The projectors onto the symmetric and antisymmetric subspaces are
$P_s(d)=\frac{1}{2}(I+\mathsf S)$, $P_a(d)=\frac{1}{2}(I-\mathsf S)$,
respectively~\cite{KlieschRoth2021}. The quantum optimal transport cost introduced in
Ref.~\cite{Friedland2022} is
\begin{equation}
    T(\rho_A,\sigma_B)
    =
    \min_{\chi_{AB}\in\mathcal B(\rho_A,\sigma_B)}
    \operatorname{Tr}
    \left[
        \chi_{AB}P_a(d)
    \right].
    \label{eq:quantum_transport_cost}
\end{equation}
Since $P_a(d)=(I-\mathsf S)/2$, this quantity gives the minimum
weight of a quantum coupling in the antisymmetric subspace.

The cost $T$ is monotone under quantum channels for qubits, but this
property does not hold in general in higher dimensions. To address
this issue, M\"uller-Hermes introduced a stabilized version of the
quantum transport cost
~\cite{mullerhermes2022monotonicityquantumoptimaltransport},
\begin{equation}
    T_s(\rho_A,\sigma_B)
    =
    \inf_{\gamma}
    T(\rho_A\otimes\gamma,\sigma_B\otimes\gamma),
    \label{eq:stabilized_transport_cost}
\end{equation}
where the infimum is taken over ancillary density operators of
arbitrary finite dimension. Following Ref.~\cite{PhysRevA.109.052443},
we also refer to $T_s$ as the revised quantum transport cost.
The stabilized cost can be evaluated using a two-dimensional
maximally mixed ancillary state,
\begin{equation}
    T_s(\rho_A,\sigma_B)
    =
    T\left(
        \rho_A\otimes\frac{I_2}{2},
        \sigma_B\otimes\frac{I_2}{2}
    \right).
    \label{eq:Ts_qubit_ancilla}
\end{equation}
In particular,
\begin{equation}
    T_s(\rho_A,\sigma_B)
    \le
    T(\rho_A,\sigma_B).
    \label{eq:Ts_le_T}
\end{equation}
Moreover, $T_s$ is unitarily invariant and monotone under quantum
channels.

Based on the revised transport cost, Shi
~\cite{PhysRevA.109.052443} introduced the coherence quantifier
\begin{equation}
    \widetilde T(\rho)
    =
    \min_{\delta\in\mathcal I}
    T_s(\rho,\delta),
    \label{eq:transport_coherence_quantifier}
\end{equation}
where $\mathcal I$ denotes the set of incoherent states with respect
to the fixed reference basis. For a pure state $\ket{\psi} = \sum_{i=1}^{d} c_i\ket{i}$,
define $p_{\max}= \max_i |c_i|^2$.
Then
\begin{equation}
    \widetilde T(\psi)
    =
    \frac{1-p_{\max}}{2}.
    \label{eq:pure_transport_coherence}
\end{equation}
Thus, $\widetilde T(\rho)$ characterizes the minimum revised
transport cost from $\rho$ to the full incoherent set. In
Sec.~\ref{sec:framework}, we instead restrict the target set to
incoherent states on the same unitary orbit as $\rho$, leading to an
orbit-restricted transport cost.

\paragraph*{Matrix notation.}
Let $S_d$ denote the symmetric group on $d$ elements, i.e., the group of
all permutations of $\{1,\ldots,d\}$~\cite{HoraObata2007}. For each $\pi\in S_d$, let
$P_\pi$ be the corresponding permutation matrix, defined by
\[
P_\pi |j\rangle = |\pi(j)\rangle,
\qquad j=1,\ldots,d.
\]
We denote the diagonal unitary subgroup by
\begin{equation*}
    \mathcal U_{\rm diag}(d)
    =
    \left\{
    \operatorname{diag}
    (e^{i\phi_1},\ldots,e^{i\phi_d})
    :
    \phi_j\in\mathbb R
    \right\}.
\end{equation*}
The group of monomial unitaries is then defined as~\cite{PhysRevLett.117.030401,FrancaHashagen2018}
\begin{equation}
    \mathcal M_d
    =
    \left\{
    P_\pi D:
    \pi\in S_d,\quad
    D\in\mathcal U_{\rm diag}(d)
    \right\}.
    \label{eq:monomial_group}
\end{equation}
Equivalently, a monomial unitary has exactly one nonzero entry of
unit modulus in each row and each column.

For a matrix $X$, we denote its Hermitian part by
\begin{equation*}
    \operatorname{Re}X
    =
    \frac{X+X^\dagger}{2}.
\end{equation*}
For a Hermitian matrix $H$, $\lambda_{\min}(H)$ denotes its smallest
eigenvalue. The operator norm is denoted by $\|\cdot\|_\infty$.


\section{Distance for unitary coherence erasure}
\label{sec:framework}

Conventional coherence distances characterize how close a quantum
state is to the full set of incoherent states. However, when the
evolution is restricted to unitary dynamics, the spectrum of the
density operator is preserved. Therefore, the incoherent state that
minimizes a conventional coherence distance may not be accessible
from the initial state by unitary evolution.

Motivated by this observation, we introduce a distance that measures
the minimum unitary cost required to reach an incoherent state with
the same spectrum. We first identify the incoherent states that can
be reached through unitary evolution and then define the distance to
this accessible set. For nondegenerate states, we show that the
optimization reduces to the distance between the eigenbasis matrix
and the group of monomial unitaries. Finally, we compare the proposed
distance with the Bures angle and the quantum optimal transport cost.

\subsection{Incoherent states on the unitary orbit}
For a density operator $\rho$, we define the set of 
unitary-accessible incoherent states by
\begin{equation}
    \mathcal I_\rho
    :=
    \mathcal I \cap \mathcal O_\rho,
    \label{eq:accessible_incoherent_set}
\end{equation}
where $\mathcal I$ is the set of incoherent states defined in
Eq.~\eqref{eq:incoherent_set}, and $\mathcal O_\rho$ is the unitary
orbit of $\rho$ defined in Eq.~\eqref{eq:unitary_orbit}.
Thus, $\mathcal I_\rho$ consists of those incoherent states that can
be reached from $\rho$ by a unitary transformation without changing
its spectrum. The structure of $\mathcal I_\rho$ can be characterized explicitly, as
the following Lemma~\ref{prop:accessible_targets}shows.
\begin{lemma}
\label{prop:accessible_targets}
Let $\rho =R\Lambda R^\dagger$, $\Lambda=\operatorname{diag} (\lambda_1,\ldots,\lambda_d).$
Then
\begin{equation}
    \mathcal I_\rho
    =
    \left\{
    P_\pi\Lambda P_\pi^\dagger:
    \pi\in S_d
    \right\}.
    \label{eq:accessible_targets}
\end{equation}
If the spectrum is nondegenerate, these are $d!$ distinct states.
In the presence of degeneracies, different permutations may
correspond to the same incoherent state.
\end{lemma}
\begin{proof}
Every state belonging to $\mathcal O_\rho$ has the same spectrum as
$\rho$. Therefore, any $\delta\in\mathcal I_\rho$ must be diagonal in the reference basis
and have eigenvalues $\{\lambda_1,\ldots,\lambda_d\}$, including their multiplicities.
Hence its diagonal entries can only differ from those of $\Lambda$ by a permutation, and therefore $\delta =P_\pi\Lambda P_\pi^\dagger$ for some $\pi\in S_d$.

Conversely, every matrix $P_\pi\Lambda P_\pi^\dagger$ is diagonal in the reference basis and
has the same spectrum as $\rho$. It therefore belongs to both $\mathcal I$ and $\mathcal O_\rho$.
\end{proof}

\subsection{Distance for unitary coherence erasure}

The restriction to $\mathcal I_\rho$ suggests a natural notion of
the cost required to erase coherence using unitary dynamics.

\begin{definition}
\label{def:unitary_erasure_distance}
The \emph{unitary coherence-erasure distance} of a density operator
$\rho$ is defined as
\begin{equation}
    D_{\mathcal I}^{U}(\rho)
    :=
    \min_{\delta\in\mathcal I_\rho}
    D_{\rm CD}(\rho,\delta).
    \label{eq:unitary_erasure_distance}
\end{equation}
\end{definition}

Here, $\mathcal I_\rho$ is the set of incoherent states on the
unitary orbit of $\rho$, introduced in
Eq.~\eqref{eq:accessible_incoherent_set}, while $D_{\rm CD}$ denotes
the quantum charging distance in Eq.~\eqref{eq:charging_distance}.
Fig.~\ref{fig:unitary_coherence_erasure_schematic} schematically
illustrates the restriction of the optimization to the accessible
incoherent subset $\mathcal I_\rho$.

\begin{figure}[t]
    \centering
    \includegraphics[width=0.95\linewidth]
    {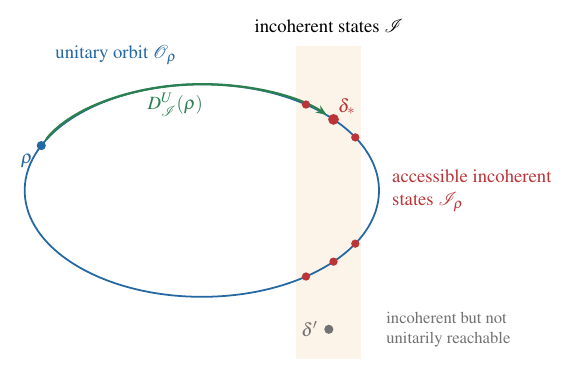}
    \caption{
    Schematic illustration of the unitary coherence-erasure distance
    $D_{\mathcal I}^{U}(\rho)$.
    The unitary orbit $\mathcal O_\rho$ intersects the incoherent set
    $\mathcal I$ at the accessible incoherent states
    $\mathcal I_\rho=\mathcal I\cap\mathcal O_\rho$.
    The optimal target $\delta_\ast\in\mathcal I_\rho$ minimizes the
    quantum charging distance from $\rho$, while incoherent states outside
    $\mathcal O_\rho$ are not unitarily reachable.
    }
    \label{fig:unitary_coherence_erasure_schematic}
\end{figure}

The quantity $D_{\mathcal I}^{U}(\rho)$ measures the minimum
unitary-control distance required to transform $\rho$ into an
incoherent state while preserving its spectrum. In contrast to a
conventional distance to the incoherent set, the optimization in
Eq.~\eqref{eq:unitary_erasure_distance} excludes incoherent target
states that are not reachable from $\rho$ through unitary evolution.

Operationally, $D_{\mathcal I}^{U}(\rho)$ quantifies the minimum
unitary-control effort required to transform $\rho$ into a state
diagonal in the fixed reference basis without altering its spectrum.
Equivalently, under a bound on the operator norm of the driving
Hamiltonian, it determines the minimum time required to eliminate
all coherence with respect to the reference basis by closed-system
unitary dynamics. This process should be distinguished from
dissipative decoherence, since the spectrum, purity, and entropy
remain unchanged throughout the unitary evolution.

\begin{lemma}
\label{prop:basic_properties}
The unitary coherence-erasure distance satisfies the following
properties:
\begin{enumerate}
    \item
    $D_{\mathcal I}^{U}(\rho)\ge0$, and  $D_{\mathcal I}^{U}(\rho)=0 \Longleftrightarrow  \rho\in\mathcal I$.
    \item
    For every monomial unitary $W\in\mathcal M_d$,
    \begin{equation}
        D_{\mathcal I}^{U}(W\rho W^\dagger)
        =
        D_{\mathcal I}^{U}(\rho).
        \label{eq:monomial_invariance}
    \end{equation}
\end{enumerate}
\end{lemma}

\begin{proof}
The first property follows directly from the nonnegativity of the
charging distance. If $\rho\in\mathcal I$, then
$\rho\in\mathcal I_\rho$, and choosing $\delta=\rho$ in
Eq.~\eqref{eq:unitary_erasure_distance} gives
$D_{\mathcal I}^{U}(\rho)=0$. Conversely, if
$D_{\mathcal I}^{U}(\rho)=0$, there exists
$\delta\in\mathcal I_\rho$ such that
$D_{\rm CD}(\rho,\delta)=0$. Since the charging distance vanishes
only when the two states coincide, one has
$\rho=\delta\in\mathcal I$.

Finally, a monomial unitary maps the incoherent set onto itself and
maps $\mathcal O_\rho$ onto
$\mathcal O_{W\rho W^\dagger}$. Together with the unitary invariance
of the charging distance, this gives
Eq.~\eqref{eq:monomial_invariance}.
\end{proof}

\begin{remark}
Here, $D_{\mathcal I}^{U}$ should be understood as a unitary-control
distance for coherence erasure, rather than as a resource-theoretic
coherence measure. It quantifies the minimum unitary effort needed to
reach an incoherent state without changing the spectrum. Hence, the
resource-theoretic monotonicity properties usually required for
coherence measures are not part of its definition.
\end{remark}

\subsection{Eigenbasis representation}

We now consider a nondegenerate state
\[
    \rho=R\Lambda R^\dagger,
    \qquad
    R\in U(d),
\]
where $\Lambda=\operatorname{diag}(\lambda_1,\ldots,\lambda_d)$,
and the columns of $R$ form an eigenbasis
$\{\ket{r_j}\}_{j=1}^d$ of $\rho$~\cite{KleinmannKampermannMeyerBruss2006}. Equivalently,
$R^\dagger = \sum_{j=1}^d \ket{j}\bra{r_j}$.

For an accessible incoherent target state
\[
    \delta_\pi
    =
    P_\pi\Lambda P_\pi^\dagger,
    \qquad
    \pi\in S_d,
\]
an eigenvector corresponding to $\lambda_j$ can be chosen as
$\ket{s_j}=P_\pi\ket{j}$.
Therefore, using the nondegenerate charging distance
Eq.~\eqref{eq:CD_nondegenerate}, we obtain
\begin{equation}
\begin{aligned}
    D_{\rm CD}(\rho,\delta_\pi)
    &=
    \min_{\phi_1,\ldots,\phi_d}
    \left\|
    i\operatorname{ln}
    \left[
        \sum_{j=1}^d
        e^{i\phi_j}
        P_\pi\ket{j}\bra{r_j}
    \right]
    \right\|_\infty
    \\
    &=
    \min_{D\in\mathcal U_{\rm diag}(d)}
    \left\|
    i\operatorname{ln}
    \left(
        P_\pi D R^\dagger
    \right)
    \right\|_\infty .
\end{aligned}
\label{eq:CD_target_phase}
\end{equation}
Here $ D= \operatorname{diag} (e^{i\phi_1},\ldots,e^{i\phi_d})$.

Recall the monomial unitary group $\mathcal M_d$ defined in
Eq.~\eqref{eq:monomial_group}. Each matrix
$M=P_\pi D\in\mathcal M_d$ contains the two degrees of freedom
appearing in Eq.~\eqref{eq:CD_target_phase}. The permutation
$P_\pi$ selects the accessible incoherent target, whereas the
diagonal unitary $D$ accounts for the arbitrary phases of its
eigenvectors. Therefore, minimizing over all accessible incoherent
targets and all eigenvector phases is equivalent to minimizing over
$\mathcal M_d$.

This motivates the following distance from the eigenbasis matrix
$R$ to the monomial unitary group:
\begin{equation}
    D_{\mathcal M}(R)
    :=
    \min_{M\in\mathcal M_d}
    \left\|
        i\operatorname{ln}
        \left(
            MR^\dagger
        \right)
    \right\|_\infty .
    \label{eq:monomial_distance}
\end{equation}
The quantity $D_{\mathcal M}(R)$ therefore measures the minimum
unitary rotation required to align the eigenbasis of $\rho$ with
the incoherent reference basis, up to permutations and phase
changes of the reference basis vectors.

\begin{lemma}
\label{prop:monomial_reduction}
For every nondegenerate state
$\rho=R\Lambda R^\dagger$, the unitary-accessible
coherence-erasure distance satisfies
\begin{equation}
    D_{\mathcal I}^{U}(\rho)
    =
    D_{\mathcal M}(R).
    \label{eq:monomial_reduction}
\end{equation}
Consequently, within the nondegenerate sector,
$D_{\mathcal I}^{U}(\rho)$ depends only on the eigenbasis of $\rho$
relative to the incoherent reference basis, and not on the numerical
values of its eigenvalues.
\end{lemma}

\begin{proof}
By Lemma~\ref{prop:accessible_targets}, the unitary-accessible
incoherent states are $\delta_\pi = P_\pi\Lambda P_\pi^\dagger,  \pi\in S_d$.
Hence, using Eq.~\eqref{eq:CD_target_phase},
$$ D_{\mathcal I}^{U}(\rho)=
    \min_{\delta_\pi\in\mathcal I_\rho}
    D_{\rm CD}(\rho,\delta_\pi)
    =
    \min_{\pi\in S_d}
    \min_{D\in\mathcal U_{\rm diag}(d)}
    \left\|
        i\operatorname{ln}
        \left(
            P_\pi D R^\dagger
        \right)
    \right\|_\infty.$$
Since every product $P_\pi D$ is an element of $\mathcal M_d$, and
conversely every $M\in\mathcal M_d$ can be written in this form, the
two minimizations can be combined into
\begin{equation*}
\begin{aligned}
    D_{\mathcal I}^{U}(\rho)=
    \min_{M\in\mathcal M_d}
    \left\|
        i\operatorname{ln}
        \left(
            MR^\dagger
        \right)
    \right\|_\infty
   = D_{\mathcal M}(R),
\end{aligned}
\end{equation*}
which proves the result.
\end{proof}

Fig.~\ref{fig:eigenbasis_reduction} summarizes this reduction
geometrically in the nondegenerate sector.

\begin{figure}[t]
    \centering
    \includegraphics[width=\linewidth]
    {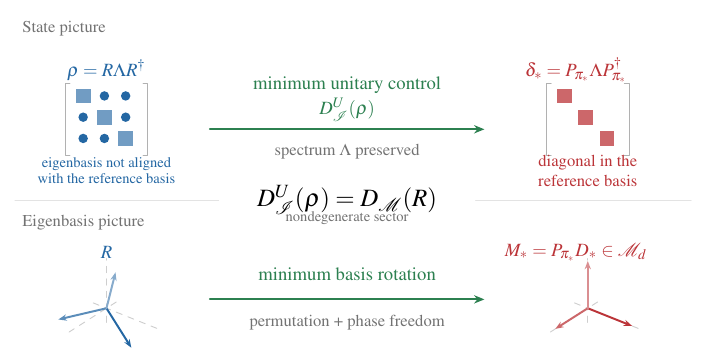}
    \caption{
    Schematic illustration of the eigenbasis reduction of unitary
    coherence erasure for a nondegenerate state.
    At the state level, the transformation from
    $\rho=R\Lambda R^\dagger$ to the optimal accessible incoherent state
    $\delta_\ast$ preserves the spectrum $\Lambda$.
    Equivalently, the problem reduces to rotating the eigenbasis matrix
    $R$ to the monomial unitary set $\mathcal M_d$, so that
    $D_{\mathcal I}^{U}(\rho)=D_{\mathcal M}(R)$.
    }
    \label{fig:eigenbasis_reduction}
\end{figure}

\begin{remark}
\label{rem:degeneracy}
The reduction in Lemma~\ref{prop:monomial_reduction} relies on
the nondegeneracy of $\rho$. When degeneracies are present, the
eigenvectors within each degenerate eigenspace can be mixed by
arbitrary unitary transformations without changing the state.
Consequently, the freedom in the charging distance optimization is
larger than the permutations and individual phase changes represented
by $\mathcal M_d$, and the above monomial unitary reduction no longer
applies directly. This additional freedom will be seen explicitly in
the qubit case discussed in Sec.~\ref{sec:qubit}.
\end{remark}
\subsection{Relations to the Bures angle and quantum transport cost}
To compare the unitary coherence-erasure distance with other
state-space quantities, we restrict the Bures angle and the revised
quantum transport cost to the same set of accessible incoherent
states $\mathcal I_\rho$.

We first define the orbit-restricted revised transport cost as
\begin{equation}
    \mathcal T_{\rm orb}(\rho)
    :=
    \min_{\delta\in\mathcal I_\rho}
    T_s(\rho,\delta).
    \label{eq:T_orb}
\end{equation}
Similarly, the orbit-restricted Bures angle is defined by
\begin{equation}
    \Theta_B^{\rm orb}(\rho)
    :=
    \min_{\delta\in\mathcal I_\rho}
    \theta_B(\rho,\delta).
    \label{eq:Bures_orb}
\end{equation}
Thus, $\mathcal T_{\rm orb}(\rho)$ and
$\Theta_B^{\rm orb}(\rho)$ measure, respectively, the minimum revised
transport cost and the minimum Bures angle from $\rho$ to an
incoherent state on the same unitary orbit. Unlike the corresponding
quantities obtained by minimizing over the full incoherent set
$\mathcal I$, both quantities incorporate the spectral constraint
imposed by unitary evolution.

Since $\mathcal I_\rho\subseteq\mathcal I$, the transport-based
coherence quantifier $\widetilde T(\rho)$ defined in
Eq.~\eqref{eq:transport_coherence_quantifier} satisfies
$\widetilde T(\rho)\le \mathcal T_{\rm orb}(\rho)$.
This follows directly because $\widetilde T(\rho)$ is minimized over
the full incoherent set, whereas $\mathcal T_{\rm orb}(\rho)$ is
minimized over the smaller set $\mathcal I_\rho$.

We next compare
$\mathcal T_{\rm orb}(\rho)$,
$\Theta_B^{\rm orb}(\rho)$, and
$D_{\mathcal I}^{U}(\rho)$.
Since all three quantities involve minimization over the same target
set $\mathcal I_\rho$, their relations can be obtained from
corresponding bounds between the underlying state-space quantities.
We begin with the revised transport cost and the Bures angle.
\begin{lemma}
\label{lem:transport_bures}
For any two density operators $\rho$ and $\sigma$ on the same
finite-dimensional Hilbert space,
\begin{equation}
    \arcsin\sqrt{2T_s(\rho,\sigma)}
    \le
    \theta_B(\rho,\sigma).
    \label{eq:pairwise_transport_bures}
\end{equation}
\end{lemma}

\begin{proof}
The quantum optimal transport cost satisfies
\begin{equation*}
    T(\rho,\sigma)
    \le
    \frac{1-F(\rho,\sigma)}{2},
\end{equation*}
where $F$ is the squared Uhlmann fidelity
~\cite{Friedland2022}. Since the revised transport cost satisfies
$T_s(\rho,\sigma)\le T(\rho,\sigma)$, we obtain
\begin{equation*}
    2T_s(\rho,\sigma)
    \le
    1-F(\rho,\sigma).
\end{equation*}
Using
$\theta_B(\rho,\sigma)=\arccos\sqrt{F(\rho,\sigma)}$, we have
\begin{equation*}
    \sqrt{1-F(\rho,\sigma)}
    =
    \sin\theta_B(\rho,\sigma).
\end{equation*}
Since $0\le\theta_B(\rho,\sigma)\le\pi/2$, it follows that
\begin{equation*}
    \arcsin\sqrt{2T_s(\rho,\sigma)}
    \le
    \theta_B(\rho,\sigma).
\end{equation*}
\end{proof}
We now apply these bounds to the common target set
$\mathcal I_\rho$. Combining them with the Bures lower bound for the
charging distance gives the following relation among the three
quantities.
\begin{theorem}
\label{thm:geometric_hierarchy}
For every finite-dimensional density operator $\rho$,
\begin{equation}  
    \arcsin
    \sqrt{
        2\mathcal T_{\rm orb}(\rho)
    }
    \le
    \Theta_B^{\rm orb}(\rho)
    \le
    D_{\mathcal I}^{U}(\rho)
    .
    \label{eq:orbit_geometric_hierarchy}
\end{equation}
\end{theorem}

\begin{proof}
By Lemma~\ref{lem:transport_bures}, for any $\arcsin\sqrt{2T_s(\rho,\delta)}\le\theta_B(\rho,\delta)$.
Since the function $x\longmapsto\arcsin\sqrt{2x}$
is monotonically increasing on the relevant interval, minimization
over the same target set $\mathcal I_\rho$ yields
\begin{equation*}
\begin{aligned}
    \arcsin
    \sqrt{
        2\mathcal T_{\rm orb}(\rho)
    }
    &=
    \arcsin
    \sqrt{
        2
        \min_{\delta\in\mathcal I_\rho}
        T_s(\rho,\delta)
    }
    \\
    &\le
    \min_{\delta\in\mathcal I_\rho}
    \theta_B(\rho,\delta)
    \\
    &=
    \Theta_B^{\rm orb}(\rho).
\end{aligned}
\end{equation*}
Furthermore, by definition every
$\delta\in\mathcal I_\rho$ belongs to the same unitary orbit as
$\rho$ and is therefore isospectral to $\rho$. Hence
Eq.~\eqref{eq:Bures_CD_bound} gives $\theta_B(\rho,\delta) \le D_{\rm CD}(\rho,\delta), \delta\in\mathcal I_\rho$.

Taking the minimum over the common target set gives
   $$ \Theta_B^{\rm orb}(\rho)
    =
    \min_{\delta\in\mathcal I_\rho}
    \theta_B(\rho,\delta)
   \le \min_{\delta\in\mathcal I_\rho}
    D_{\rm CD}(\rho,\delta)
= D_{\mathcal I}^{U}(\rho),$$
which completes the proof.
\end{proof}
The two inequalities in
Theorem~\ref{thm:geometric_hierarchy} become simultaneously tight
for pure states. In this case, the unitary-accessible incoherent set
has a particularly simple form, allowing all three quantities to be
evaluated explicitly.
For pure states, the above hierarchy becomes an equality, leading to
the following Corollary~\ref{cor:pure_state_equality}.
\begin{corollary}
\label{cor:pure_state_equality}
Let  $\ket{\psi}
    =
    \sum_{i=1}^{d}
    \lambda_i\ket{i},
    p_{\max}
    :=
    \max_i |\lambda_i|^2$.
Then
\begin{equation}
    \arcsin\sqrt{2\mathcal T_{\rm orb}(\psi)}
    =
    \Theta_B^{\rm orb}(\psi)
    =
    D_{\mathcal I}^{U}(\psi)
    =
    \arccos\sqrt{p_{\max}}.
    \label{eq:pure_state_distance_equality}
\end{equation}
Moreover,
\begin{equation}
    \mathcal T_{\rm orb}(\psi)
    =
    \widetilde T(\psi)
    =
    \frac{1-p_{\max}}{2}.
    \label{eq:pure_state_transport_equality}
\end{equation}
\end{corollary}

\begin{proof}
For the pure state $\ket{\psi}$, the spectrum is
$(1,0,\ldots,0)$. Hence, the incoherent states on the same unitary
orbit are precisely the reference-basis projectors,
$\mathcal I_\psi =\left\{ \ket{i}\bra{i}: i=1,\ldots,d \right\}$.

Let $i_0$ satisfy
$|\lambda_{i_0}|^2=p_{\max}$, and define
$\delta_0=\ket{i_0}\bra{i_0}$.
Since $\mathcal I_\psi\subseteq\mathcal I$, $\widetilde T(\psi)\le\mathcal T_{\rm orb}(\psi)$.
On the other hand, $\mathcal T_{\rm orb}(\psi) \le T_s(\psi,\delta_0)\le T(\psi,\delta_0).$

Because both $\psi$ and $\delta_0$ are pure states, their only
quantum coupling is the product state
$\psi\otimes\delta_0$. Therefore,
\begin{equation*}
\begin{aligned}
    T(\psi,\delta_0)
    =
    \operatorname{Tr}
    \left[
        (\psi\otimes\delta_0)P_a
    \right] 
    =
    \frac{
        1-\operatorname{Tr}(\psi\delta_0)
    }{2}
    =
    \frac{1-p_{\max}}{2}.
\end{aligned}
\end{equation*}
Together with $\widetilde T(\psi)=\frac{1-p_{\max}}{2}$,
we obtain
\begin{equation*}
    \mathcal T_{\rm orb}(\psi)
    =
    \widetilde T(\psi)
    =
    \frac{1-p_{\max}}{2}.
\end{equation*}

Next, minimizing the Bures angle over $\mathcal I_\psi$ gives
\begin{equation*}
\begin{aligned}
    \Theta_B^{\rm orb}(\psi)
    =
    \min_i
    \arccos|\langle i|\psi\rangle|
    =
    \arccos\sqrt{p_{\max}}.
\end{aligned}
\end{equation*}
For pure states, the charging distance coincides with the Bures
angle. Hence, $ D_{\mathcal I}^{U}(\psi)  = \min_i D_{\rm CD} \left( \psi,\ket{i}\bra{i} \right)
 =\arccos\sqrt{p_{\max}}.$
Finally, $$\arcsin\sqrt{2\mathcal T_{\rm orb}(\psi)} = \arcsin\sqrt{1-p_{\max}} = \arccos\sqrt{p_{\max}},$$
which completes the proof.
\end{proof}
Theorem~\ref{thm:geometric_hierarchy} places the three constructions
within a common unitary-accessible framework. The stabilized
transport cost characterizes proximity through quantum couplings,
the Bures angle measures state-space distinguishability, whereas
$D_{\mathcal I}^{U}$ additionally incorporates the requirement that
the incoherent target be reached through a spectrum-preserving
unitary transformation.

For pure states, the rescaled transport quantity
$\arcsin\sqrt{2\mathcal T_{\rm orb}}$, the orbit-restricted Bures
angle, and the unitary-accessible coherence-erasure distance coincide
exactly. For mixed states, these quantities need not coincide,
because state-space proximity and unitary accessibility impose
different geometric constraints. In the next section, we analyze
this distinction explicitly for qubit states.

\section{Unitary coherence erasure for qubits}
\label{sec:qubit}
We now consider qubit states, for which the unitary coherence-erasure
distance can be obtained exactly. The Bloch-sphere representation
gives a simple geometric picture and highlights the difference
between unitary and state-space distances for mixed states, as well
as the role of spectral degeneracy.
\subsection{Exact distance}
Consider a qubit state in Bloch form~\cite{LiuSunXuZhao2026},
\begin{equation*}
    \rho
    =
    \frac{1}{2}
    \left(
        I+\mathbf r\cdot\boldsymbol{\sigma}
    \right),
    \qquad
    \mathbf r=(r_x,r_y,r_z),
\end{equation*}
where
$r:=\|\mathbf r\|=\sqrt{r_x^2+r_y^2+r_z^2}\le 1$.
Its eigenvalues are
\[
\lambda_\pm=\frac{1\pm r}{2}.
\]
For $r>0$, the eigenvalues are different, so $\rho$ is
nondegenerate. Since unitary evolution preserves the spectrum, the
accessible incoherent states are obtained by placing
$\lambda_+$ and $\lambda_-$ on the diagonal in the reference basis.
This leads directly to the following proposition.

\begin{proposition}
\label{prop:qubit_accessible_targets}
For a nondegenerate qubit state with $r>0$, the set of accessible
incoherent states is
\begin{equation}
    \mathcal I_\rho
    =
    \left\{
        \frac12(I+r\sigma_z),
        \frac12(I-r\sigma_z)
    \right\}.
    \label{eq:qubit_accessible_targets}
\end{equation}
\end{proposition}

\begin{proof}
An incoherent qubit has its Bloch vector on the $z$ axis, while
unitary evolution preserves the Bloch radius $r$. Hence, the unitary
orbit of $\rho$ intersects the incoherent axis only at
$\pm r\mathbf e_z$, which gives
Eq.~\eqref{eq:qubit_accessible_targets}.
\end{proof}

We can now minimize the charging distance over these two accessible
incoherent states.

\begin{theorem}
\label{thm:qubit_exact_distance}
Let $\rho= \frac12\left(   I+\mathbf r\cdot\boldsymbol{\sigma}\right)$,
$ \mathbf r=(r_x,r_y,r_z)$,
be a nondegenerate qubit state with $r=\|\mathbf r\|>0$. Then
\begin{equation}
    D_{\mathcal I}^{U}(\rho)
    =
    \frac12
    \arccos
    \left(
        \frac{|r_z|}{r}
    \right).
    \label{eq:qubit_exact_distance}
\end{equation}
Equivalently,
\begin{equation}
    D_{\mathcal I}^{U}(\rho)
    =
    \frac12
    \arcsin
    \left(
        \frac{r_\perp}{r}
    \right),
    \qquad
    r_\perp=\sqrt{r_x^2+r_y^2}.
    \label{eq:qubit_exact_distance_perp}
\end{equation}
\end{theorem}

\begin{proof}
By Proposition~\ref{prop:qubit_accessible_targets}, the two accessible
incoherent states are $\delta_\pm  =  \frac12(I\pm r\sigma_z)$.
Let $\mathbf n=\mathbf r/r$ be the direction of the Bloch vector.
For two nondegenerate isospectral qubit states with Bloch directions
$\mathbf n$ and $\mathbf m$, the charging distance is
\begin{equation}
    D_{\rm CD}(\mathbf n,\mathbf m)
    =
    \frac12
    \arccos(\mathbf n\cdot\mathbf m).
    \label{eq:qubit_CD_angle}
\end{equation}
Therefore,
\begin{equation*}
    D_{\rm CD}(\rho,\delta_\pm)
    =
    \frac12
    \arccos
    \left(
        \pm\frac{r_z}{r}
    \right).
\end{equation*}
Taking the smaller of the two values gives
\begin{equation*}
\begin{aligned}
    D_{\mathcal I}^{U}(\rho)
    &=
    \frac12
    \min
    \left\{
        \arccos\frac{r_z}{r},
        \arccos\left(-\frac{r_z}{r}\right)
    \right\}
    \\
    &=
    \frac12
    \arccos
    \left(
        \frac{|r_z|}{r}
    \right),
\end{aligned}
\end{equation*}
where we used the monotonicity of $\arccos$ on $[-1,1]$.

Finally, since
\[
    \frac{r_\perp}{r}
    =
    \sqrt{1-\frac{r_z^2}{r^2}},
\]
we obtain
\[
    \arccos\frac{|r_z|}{r}
    =
    \arcsin\frac{r_\perp}{r},
\]
which proves Eq.~\eqref{eq:qubit_exact_distance_perp}.
\end{proof}

Eq.~\eqref{eq:qubit_exact_distance} has a simple geometric
meaning. Unitary coherence erasure rotates the Bloch vector onto the
reference $z$ axis while preserving its length. Hence,
$D_{\mathcal I}^{U}$ is one half of the smallest Bloch-sphere
rotation angle needed to reach the $z$ axis.

The distance depends only on the direction $\mathbf r/r$ and not on
the Bloch radius $r$. Thus, for all nondegenerate qubit states,
$D_{\mathcal I}^{U}$ measures the mismatch between the eigenbasis of
$\rho$ and the reference basis rather than the mixedness of the
state. This is consistent with the eigenvalue independence established
in Lemma~\ref{prop:monomial_reduction}.

\subsection{Comparison with the Bures angle}

The exact expression obtained above shows that the
unitary-accessible coherence-erasure distance of a nondegenerate
qubit depends only on the direction of its Bloch vector. It is then
natural to compare this unitary-control distance with the
orbit-restricted Bures angle introduced in
Eq.~\eqref{eq:Bures_orb}. Unlike
$D_{\mathcal I}^{U}(\rho)$, the Bures angle is generally sensitive
to both the direction and the length of the Bloch vector.

For two qubit states
$ \rho  =  \frac12 \left(    I+\mathbf r\cdot\boldsymbol{\sigma}\right)$,
 $ \sigma= \frac12 \left(    I+\mathbf s\cdot\boldsymbol{\sigma} \right)$,
their squared Uhlmann fidelity can be written as
\begin{equation}
    F(\rho,\sigma)
    =
    \frac12
    \left[
        1+\mathbf r\cdot\mathbf s
        +
        \sqrt{
            (1-\|\mathbf r\|^2)
            (1-\|\mathbf s\|^2)
        }
    \right].
    \label{eq:qubit_fidelity}
\end{equation}

For the state considered in the previous subsection,
the two unitary-accessible incoherent targets are
$\delta_\pm= \frac12\left(I\pm r\sigma_z\right)$,
whose Bloch vectors are  $\mathbf s_\pm =\pm r\mathbf e_z$.
Since $\|\mathbf s_\pm\|=\|\mathbf r\|=r$, substitution into
Eq.~\eqref{eq:qubit_fidelity} gives
\begin{equation*}
    F(\rho,\delta_\pm)
    =
    \frac12
    \left[
        1
        \pm rr_z
        +(1-r^2)
    \right]
   =
    1-\frac{r^2}{2}
    \pm\frac{rr_z}{2}.
\label{eq:qubit_target_fidelity}
\end{equation*}
The Bures angle is minimized by the target having the larger
fidelity with $\rho$. Therefore,
\begin{equation}
    \max_{\delta\in\mathcal I_\rho}
    F(\rho,\delta)
    =
    1-\frac{r^2-r|r_z|}{2}.
    \label{eq:qubit_max_fidelity}
\end{equation}
We thus obtain the following closed-form expression.
\begin{proposition}
\label{prop:qubit_bures}
For every nondegenerate qubit state
\begin{equation*}
    \rho
    =
    \frac12
    \left(
        I+\mathbf r\cdot\boldsymbol{\sigma}
    \right),
    \qquad
    r=\|\mathbf r\|>0,
\end{equation*}
the orbit-restricted Bures angle is
\begin{equation}
    \Theta_B^{\rm orb}(\rho)
    =
    \arccos
    \sqrt{
        1-\frac{r^2-r|r_z|}{2}
    }.
    \label{eq:qubit_bures}
\end{equation}
Equivalently,
\begin{equation}
    \Theta_B^{\rm orb}(\rho)
    =
    \frac12
    \arccos
    \left(
        1-r^2+r|r_z|
    \right).
    \label{eq:qubit_bures_halfangle}
\end{equation}
\end{proposition}
\begin{proof}
By definition of the orbit-restricted Bures angle~\eqref{eq:Bures_orb}
\begin{equation*}
    \Theta_B^{\rm orb}(\rho)
    =
    \min_{\delta\in\mathcal I_\rho}
    \arccos\sqrt{F(\rho,\delta)}.
\end{equation*}
Since $\arccos x$ is monotonically decreasing on $[0,1]$, this
minimization is equivalent to maximizing the fidelity over
$\mathcal I_\rho$. Using
Eq.~\eqref{eq:qubit_max_fidelity}, we immediately obtain
\begin{equation*}
    \Theta_B^{\rm orb}(\rho)
    =
    \arccos
    \sqrt{
        1-\frac{r^2-r|r_z|}{2}
    }.
\end{equation*}
Note that
\begin{equation*}
    1-\frac{r^2-r|r_z|}{2}
    =
    \frac{
        1+\left(1-r^2+r|r_z|\right)
    }{2}.
\end{equation*}
Using the half-angle identity $\cos^2\frac{\theta}{2} =\frac{1+\cos\theta}{2}$,
together with  $0\le\Theta_B^{\rm orb}\le\pi/2$, gives Eq.~\eqref{eq:qubit_bures_halfangle}.
\end{proof}
Eq.~\eqref{eq:qubit_bures_halfangle} makes the contrast with
the unitary-accessible coherence-erasure distance particularly
transparent. While
$D_{\mathcal I}^{U}(\rho)$ depends only on the ratio
$|r_z|/r$, and hence only on the direction of the Bloch vector,
$\Theta_B^{\rm orb}(\rho)$ also depends explicitly on its length
$r$. The latter therefore retains information about the mixedness
of the state. As shown below, this difference leads to a strict
separation between the two quantities for mixed coherent qubit
states.


\subsection{Accessibility difference}
\label{subsec:qubit_degeneracy}

The preceding results show that, for a mixed coherent qubit, the
orbit-restricted Bures angle is strictly smaller than the
unitary-accessible coherence-erasure distance. To quantify this
difference, we define the \emph{accessibility difference} as
\begin{equation}
    \Delta_{\rm acc}(\rho)
    :=
    D_{\mathcal I}^{U}(\rho)
    -
    \Theta_B^{\rm orb}(\rho).
    \label{eq:accessibility_difference}
\end{equation}
For a nondegenerate qubit state, $r>0$,
Eqs.~\eqref{eq:qubit_exact_distance} and
\eqref{eq:qubit_bures_halfangle} give
\begin{equation}
    \Delta_{\rm acc}(\rho)
    =
    \frac12
    \left[
        \arccos\left(\frac{|r_z|}{r}\right)
        -
        \arccos\left(1-r^2+r|r_z|\right)
    \right].
    \label{eq:qubit_accessibility_difference}
\end{equation}
Thus, $\Delta_{\rm acc}$ quantifies the difference between the
minimum unitary-control distance required to reach an incoherent
state and the corresponding proximity measured by the Bures
geometry.
It is useful to examine this difference in two limiting regimes.

We first consider the pure-state case, $r=1$. In this limit,
\begin{equation*}
    1-r^2+r|r_z|
    =
    |r_z|,
\end{equation*}
and therefore
\begin{equation}
    \Theta_B^{\rm orb}(\rho)
    =
    D_{\mathcal I}^{U}(\rho)
    =
    \frac12\arccos|r_z|.
    \label{eq:qubit_pure_distance}
\end{equation}
For a pure qubit, the largest population in the reference basis is
\begin{equation*}
    p_{\max}
    =
    \frac{1+|r_z|}{2}.
\end{equation*}
Using the half-angle identity,
Eq.~\eqref{eq:qubit_pure_distance} can equivalently be written as
\begin{equation*}
    \frac12\arccos|r_z|
    =
    \arccos
    \sqrt{\frac{1+|r_z|}{2}}
    =
    \arccos\sqrt{p_{\max}}.
\end{equation*}
Hence
$ \Delta_{\rm acc}(\rho)= 0$, $ r=1$,
in agreement with the pure-state saturation obtained in
Corollary~\ref{cor:pure_state_equality}.

A qualitatively different behavior appears near the maximally mixed
state. To see this, consider a family of nondegenerate qubit states
approaching the center of the Bloch ball while keeping the direction
of the Bloch vector fixed. Introducing
\begin{equation*}
    a
    :=
    \frac{|r_z|}{r},
    \qquad
    0\le a<1,
\end{equation*}
we obtain
\begin{equation*}
    D_{\mathcal I}^{U}(\rho)
    =
    \frac12\arccos a,
\end{equation*}
which is independent of the Bloch radius $r$. By contrast,
\begin{equation*}
    \Theta_B^{\rm orb}(\rho)
    =
    \frac12
    \arccos
    \left[
        1-r^2(1-a)
    \right].
    \label{eq:qubit_bures_fixed_direction}
\end{equation*}
Consequently, $\lim_{r\to0^+} \Theta_B^{\rm orb}(\rho) = 0$,
whereas $\lim_{r\to0^+} D_{\mathcal I}^{U}(\rho)  = \frac12\arccos a$.
More explicitly, using
$\arccos(1-x)\sim\sqrt{2x}$ as $x\to0^+$, we find
\begin{equation*}
    \Theta_B^{\rm orb}(\rho)
    \sim
    r\sqrt{\frac{1-a}{2}},
    \qquad
    r\to0^+.
    \label{eq:qubit_bures_origin_asymptotic}
\end{equation*}
The accessibility difference therefore approaches
\begin{equation*}
    \lim_{r\to0^+}
    \Delta_{\rm acc}(\rho)
    =
    \frac12\arccos a.
    \label{eq:qubit_difference_origin_limit}
\end{equation*}
Thus, along every fixed coherent direction $a<1$, the Bures
distance to the accessible incoherent set vanishes as the state
approaches the center of the Bloch ball, whereas the corresponding
unitary-control distance remains finite.

At the exactly maximally mixed state   $\rho_\ast =\frac{I}{2}$,
the situation changes qualitatively. The state is itself incoherent,
and its unitary orbit consists only of the state itself. Hence
\begin{equation*}
    D_{\mathcal I}^{U}
    \left(
        \frac{I}{2}
    \right)
    =
    \Theta_B^{\rm orb}
    \left(
        \frac{I}{2}
    \right)
    =
    0,
    \label{eq:maximally_mixed_zero}
\end{equation*}
and consequently
\begin{equation*}
    \Delta_{\rm acc}
    \left(
        \frac{I}{2}
    \right)
    =
    0.
\end{equation*}
Therefore, for a family approaching $I/2$ along a fixed coherent
Bloch direction,
\begin{equation*}
    \lim_{r\to0^+}
    D_{\mathcal I}^{U}(\rho)
    \neq
    D_{\mathcal I}^{U}
    \left(
        \frac{I}{2}
    \right),
    \label{eq:qubit_degeneracy_discontinuity}
\end{equation*}
and likewise
\begin{equation*}
    \lim_{r\to0^+}
    \Delta_{\rm acc}(\rho)
    \neq
    \Delta_{\rm acc}
    \left(
        \frac{I}{2}
    \right).
\end{equation*}
This singular behavior originates from the change in spectral
multiplicity at $r=0$. For every $r>0$, the two eigenvalues are
distinct, and the direction of the Bloch vector determines a
distinguished eigenbasis. A finite rotation of this eigenbasis may
therefore be required even arbitrarily close to the maximally mixed
state. At $r=0$, by contrast, the state is fully degenerate and every
orthonormal basis is an eigenbasis. The eigenbasis mismatch with
respect to the incoherent reference basis then ceases to distinguish
different representations of the state, and the unitary-accessible
coherence-erasure distance collapses to zero.

The qubit case therefore reveals two complementary aspects of the
geometry introduced in this work. Within the nondegenerate sector,
$D_{\mathcal I}^{U}$ depends only on the orientation of the
eigenbasis relative to the incoherent reference basis and is
independent of the eigenvalues, whereas the orbit-restricted Bures
angle remains sensitive to the mixedness of the state. Their
difference is captured by $\Delta_{\rm acc}$. At spectral
degeneracy, the enlargement of the eigenspace freedom changes this
geometry discontinuously. In higher dimensions, the same mechanism
is accompanied by the additional permutation and phase freedoms
encoded by the monomial unitary group. We next investigate this
richer structure for qutrit states.



\section{Unitary coherence erasure for qutrits}
\label{sec:qutrit}
We now consider qutrit states. Unlike the qubit case, the problem
cannot be described by a simple rotation on the Bloch sphere.
Instead, the optimization involves both permutations of the
eigenvalues and phases of the eigenvectors.

Using the monomial unitary formulation developed above, we reduce the
problem to a finite-dimensional matrix optimization. We then derive a
tight universal upper bound on the unitary coherence-erasure distance
and give an explicit eigenbasis that attains the bound.

\subsection{Qutrit formulation}
\label{subsec:qutrit_formulation}

Let
\begin{equation*}
    \rho
    =
    R\Lambda R^\dagger,
    \qquad
    R\in U(3),
\end{equation*}
be a nondegenerate qutrit state, where
\begin{equation*}
    \Lambda
    =
    \operatorname{diag}(\lambda_1,\lambda_2,\lambda_3),
    \qquad
    \lambda_i\neq\lambda_j
    \quad (i\neq j).
\end{equation*}
The columns of $R$ form an eigenbasis of $\rho$ relative to the
fixed reference basis.

As shown in Sec.~\ref{sec:framework}, the accessible incoherent
states are obtained by permuting the eigenvalues along the reference
basis. The remaining freedom is given by the phases of the
eigenvectors. These two freedoms are combined in the monomial
unitary group $\mathcal M_3$. Therefore, Lemma~\ref{prop:monomial_reduction} gives
\begin{equation}
    D_{\mathcal I}^{U}(\rho)
    =
    \min_{M\in\mathcal M_3}
    \left\|
        i\operatorname{ln}
        \left(
            MR^\dagger
        \right)
    \right\|_\infty.
    \label{eq:qutrit_monomial_distance}
\end{equation}

Eq.~\eqref{eq:qutrit_monomial_distance} shows that, for
nondegenerate qutrit states, the distance depends on the eigenbasis
$R$ but not on the eigenvalues
$\lambda_1,\lambda_2,\lambda_3$. Thus, the problem is determined by
the orientation of the eigenbasis relative to the reference basis.

To simplify Eq.~\eqref{eq:qutrit_monomial_distance}, define
\begin{equation*}
    \mu(R)
    :=
    \max_{M\in\mathcal M_3}
    \lambda_{\min}
    \left[
        \operatorname{Re}
        \left(
            MR^\dagger
        \right)
    \right].
\end{equation*}

The following lemma gives an equivalent expression for the distance.

\begin{lemma}
\label{lem:qutrit_mu}
For every nondegenerate qutrit state
$\rho=R\Lambda R^\dagger$,
\begin{equation}
    D_{\mathcal I}^{U}(\rho)
    =
    \arccos\mu(R).
    \label{eq:qutrit_distance_mu}
\end{equation}
\end{lemma}

\begin{proof}
For a fixed $M\in\mathcal M_3$, let $W=MR^\dagger$. Since $W$ is
unitary, it can be written as
\begin{equation*}
    W
    =
    V\operatorname{diag}
    \left(
        e^{i\theta_1},
        e^{i\theta_2},
        e^{i\theta_3}
    \right)V^\dagger,
    \qquad
    \theta_j\in(-\pi,\pi].
\end{equation*}
Hence,
\begin{equation*}
\begin{aligned}
    \left\|i\operatorname{ln}W\right\|_\infty
    &=
    \max_j|\theta_j|,
    \\
    \lambda_{\min}(\operatorname{Re}W)
    &=
    \min_j\cos\theta_j
    =
    \cos\left(\max_j|\theta_j|\right).
\end{aligned}
\end{equation*}
Therefore,
\begin{equation*}
    \left\|i\operatorname{ln}W\right\|_\infty
    =
    \arccos
    \left[
        \lambda_{\min}(\operatorname{Re}W)
    \right].
\end{equation*}
Substituting $W=MR^\dagger$ into
Eq.~\eqref{eq:qutrit_monomial_distance} and using the monotonicity
of $\arccos$ gives
\begin{equation*}
\begin{aligned}
    D_{\mathcal I}^{U}(\rho)
    &=
    \arccos
    \left\{
        \max_{M\in\mathcal M_3}
        \lambda_{\min}
        \left[
            \operatorname{Re}
            \left(
                MR^\dagger
            \right)
        \right]
    \right\}
    \\
    &=
    \arccos\mu(R).
\end{aligned}
\end{equation*}
\end{proof}

Lemma~\ref{lem:qutrit_mu} replaces the matrix logarithm by an
optimization involving the smallest eigenvalue. Therefore, to obtain
a universal upper bound on $D_{\mathcal I}^{U}$, it is enough to find
a constant $c$ such that, for every $R\in U(3)$, there exists
$M\in\mathcal M_3$ satisfying
\begin{equation}
    \lambda_{\min}
    \left[
        \operatorname{Re}
        \left(
            MR^\dagger
        \right)
    \right]
    \ge c.
    \label{eq:qutrit_covering_problem}
\end{equation}
The largest such constant determines the largest possible value of
the unitary coherence erasure distance for qutrits. We determine it
in the next subsection.

\subsection{A tight bound}
\label{subsec:qutrit_tight_bound}

By Lemma~\ref{lem:qutrit_mu},  $ D_{\mathcal I}^{U}(\rho)  =\arccos\mu(R),   \rho=R\Lambda R^\dagger$.
Hence, the largest possible distance is determined by the smallest
value of $\mu(R)$. Define
\begin{equation}
    \mu_\ast
    :=
    \inf_{R\in U(3)}\mu(R).
    \label{eq:qutrit_mu_star}
\end{equation}
Then $\sup_{R\in U(3)} D_{\mathcal I}^{U}(R\Lambda R^\dagger)  = \arccos\mu_\ast$.
For qutrits, the extremal value $\mu_\ast$ can be determined exactly. Let
\[
    c_\ast
    :=
    \frac{2\sqrt2-1}{4},
    \qquad
    \theta_\ast
    :=
    \arccos c_\ast.
\]

\begin{theorem}
\label{thm:qutrit_tight_bound}
For every nondegenerate qutrit state
$\rho=R\Lambda R^\dagger$,
\begin{equation}
    D_{\mathcal I}^{U}(\rho)
    \le
    \theta_\ast
    =
    \arccos
    \left(
        \frac{2\sqrt2-1}{4}
    \right).
    \label{eq:qutrit_tight_bound}
\end{equation}
Equivalently,
\[
    \mu(R)
    \ge
    c_\ast
    =
    \frac{2\sqrt2-1}{4},
    \qquad
    R\in U(3).
\]
The bound is tight.
Numerically,
$c_\ast\simeq0.4571068$ and
$\theta_\ast\simeq1.09606~{\rm rad}\simeq62.80^\circ$.
\end{theorem}

\begin{proof}
The proof is given in Appendix~\ref{app:qutrit_bound}.
\end{proof}

Therefore,
\[
    \sup_{\rho\ {\rm nondegenerate}}
    D_{\mathcal I}^{U}(\rho)
    =
    \theta_\ast.
\]
This value is independent of the eigenvalues of $\rho$ and depends
only on the orientation of its eigenbasis relative to the reference
basis.

An explicit eigenbasis attaining the bound is given in the next subsection.

\subsection{A qutrit eigenbasis saturating the bound}
\label{subsec:qutrit_extremal}
The tight bound obtained in the previous subsection is attained by
an explicit qutrit eigenbasis. Consider
\begin{equation*}
R_\ast
=
\begin{pmatrix}
\dfrac1{\sqrt2} & 0 & \dfrac1{\sqrt2}\\[2mm]
-\dfrac12 & \dfrac1{\sqrt2} & \dfrac12\\[2mm]
-\dfrac12 &-\dfrac1{\sqrt2} & \dfrac12
\end{pmatrix}.
\end{equation*}
A direct calculation gives  $ R_\ast^T R_\ast=I, \det R_\ast=1$,
and hence $R_\ast\in SO(3)\subset U(3)$.

After a permutation of the second and third columns, its modulus
matrix takes the form
\begin{equation*}
    |R_\ast|
    \sim
    \begin{pmatrix}
        \dfrac1{\sqrt2}
        &
        \dfrac1{\sqrt2}
        &
        0
        \\[2mm]
        \dfrac12
        &
        \dfrac12
        &
        \dfrac1{\sqrt2}
        \\[2mm]
        \dfrac12
        &
        \dfrac12
        &
        \dfrac1{\sqrt2}
    \end{pmatrix}.
\end{equation*}
This matrix coincides with the extremal endpoint identified in
Appendix~\ref{app:qutrit_bound}.
The following proposition makes the saturation of the universal
bound explicit.

\begin{proposition}
\label{prop:qutrit_extremal}
Let
\begin{equation*}
    \Lambda
    =
    \operatorname{diag}
    (\lambda_1,\lambda_2,\lambda_3),
    \qquad
    \lambda_i\neq\lambda_j
    \quad (i\neq j),
\end{equation*}
and define $\rho_\ast:= R_\ast\Lambda R_\ast^\dagger$.
Then $\mu(R_\ast)= c_\ast =\frac{2\sqrt2-1}{4}$,
and therefore
\begin{equation*}
    D_{\mathcal I}^{U}(\rho_\ast)
    =
    \theta_\ast
    =
    \arccos
    \left(
        \frac{2\sqrt2-1}{4}
    \right).
\end{equation*}
Thus the universal qutrit bound in
Theorem~\ref{thm:qutrit_tight_bound} is attained for every
nondegenerate spectrum when the eigenbasis is $R_\ast$.
\end{proposition}

\begin{proof}
The exact monomial unitary optimization for $R_\ast$ is established
in Appendix~\ref{app:exact_extremality}, where it is shown that $\mu(R_\ast)=c_\ast$.
Lemma~\ref{lem:qutrit_mu} then gives
\begin{equation*}
    D_{\mathcal I}^{U}(\rho_\ast)=\arccos\mu(R_\ast)
    = \arccos c_\ast=\theta_\ast.
\end{equation*}
Since $\mu(R)$ depends only on the eigenbasis in the nondegenerate
sector, the result is independent of the particular values of
$\lambda_1,\lambda_2,\lambda_3$.
\end{proof}

There is also a simple geometric feature of the extremal basis.
Since $R_\ast\in SO(3)$, it represents a three-dimensional rotation.
Moreover, $ \operatorname{tr}R_\ast =\sqrt2+\frac12 = 1+2c_\ast$.
Hence the rotation angle $\vartheta_\ast$ of $R_\ast$ satisfies
    $\cos\vartheta_\ast
    =
    \frac{\operatorname{tr}R_\ast-1}{2}
    =
    c_\ast$,
and therefore
$\vartheta_\ast  = \theta_\ast$.
 
For the identity monomial unitary, the relative unitary
$R_\ast^\dagger$ already has eigenphases
$0,\pm\theta_\ast$. The nontrivial part of the extremality result is
that no other permutation and phase choice in $\mathcal M_3$ can
reduce the maximal absolute eigenphase below $\theta_\ast$.

The extremality is therefore a property of the eigenbasis geometry
rather than of a particular spectrum. Any nondegenerate spectrum
combined with $R_\ast$ gives the same maximal unitary-accessible
coherence-erasure distance.

\subsection{Geometric interpretation of the qutrit bound}
\label{subsec:qutrit_geometric}
The qutrit result admits a natural geometric interpretation.
Eq.~\eqref{eq:qutrit_monomial_distance} expresses the
coherence-erasure problem as the distance from the eigenbasis matrix
$R$ to the monomial unitary subgroup $\mathcal M_3$,
   $ D_{\mathcal I}^{U}(\rho)
    =
    \min_{M\in\mathcal M_3}
    \left\|
        i\operatorname{Ln}(MR^\dagger)
    \right\|_\infty $.
Thus, in the nondegenerate sector, the tight constant
$\theta_\ast$ is the largest possible distance of a qutrit
eigenbasis from $\mathcal M_3$,
\begin{equation*}
\theta_\ast
    =
    \sup_{R\in U(3)}
    \min_{M\in\mathcal M_3}
    \left\|
        i\operatorname{ln}(MR^\dagger)
    \right\|_\infty .
    \label{eq:qutrit_geometric_radius}
\end{equation*}
In this sense, Theorem~\ref{thm:qutrit_tight_bound} determines the
covering radius associated with the monomial unitary reduction in
dimension three.

This also provides a direct comparison with the qubit case. For a
nondegenerate qubit with Bloch vector
$\boldsymbol r=(r_x,r_y,r_z)$, the corresponding distance is
    $D_{\mathcal I}^{U}(\rho)
    =
    \frac12
    \arccos
    \frac{|r_z|}{|\boldsymbol r|}$,
and therefore
\begin{equation*}
    \sup_{\rho\ {\rm nondegenerate}}
    D_{\mathcal I}^{U}(\rho)
    =
    \frac{\pi}{4}
    =
    45^\circ,
\end{equation*}
with the maximum attained when the Bloch vector lies in the
equatorial plane. For qutrits, Theorem~\ref{thm:qutrit_tight_bound}
instead gives
\begin{equation*}
    \sup_{\rho\ {\rm nondegenerate}}
    D_{\mathcal I}^{U}(\rho)
    =
    \theta_\ast
    =
    \arccos
    \left(
        \frac{2\sqrt2-1}{4}
    \right)
    \simeq
    62.80^\circ.
\end{equation*}
Hence $\frac{\pi}{4}<\theta_\ast$.

The qutrit case allows a larger maximal coherence-erasure distance
than the qubit case. In other words, some qutrit eigenbases require
a larger unitary rotation to reach the corresponding incoherent
reference basis than is possible for qubits. This difference arises
from the additional basis degrees of freedom in dimension three.
We do not claim that the maximal distance increases monotonically
with the Hilbert space dimension.

It is important, however, to distinguish this nondegenerate geometry
from the degenerate case. If the spectrum contains degeneracies, the
eigenbasis matrix $R$ is no longer unique, since arbitrary unitary
rotations may be performed within each degenerate eigenspace. More
precisely, for $\rho=R\Lambda R^\dagger$,
introduce the stabilizer of the spectrum,
\begin{equation*}
    \mathcal G_\Lambda
    :=
    \left\{
        V\in U(3):
        V\Lambda V^\dagger=\Lambda
    \right\}.
    \label{eq:qutrit_spectral_stabilizer}
\end{equation*}
The corresponding coherence-erasure optimization then takes the form
\begin{equation}
    D_{\mathcal I}^{U}(\rho)
    =
    \min_{\substack{\pi\in S_3\\ V\in\mathcal G_\Lambda}}
    \left\|
        i\operatorname{ln}
        \left(
            P_\pi V R^\dagger
        \right)
    \right\|_\infty ,
    \label{eq:qutrit_degenerate_reduction}
\end{equation}
with redundant permutations understood when eigenvalues coincide.

For a nondegenerate spectrum,
    $\mathcal G_\Lambda
    =
    \mathcal U_{\rm diag}(3)$,
and Eq.~\eqref{eq:qutrit_degenerate_reduction} reduces to the
monomial unitary formulation used throughout this section. By
contrast, if, for example,
$\lambda_1=\lambda_2\neq\lambda_3$, then
   $\mathcal G_\Lambda
    \simeq
    U(2)\times U(1)$,
so that additional rotations inside the two-dimensional degenerate
eigenspace become available.

\begin{remark}
The tight qutrit bound above applies only to nondegenerate states.
When the spectrum is degenerate, the eigenvectors inside a degenerate
eigenspace are not uniquely defined, and they can be mixed by unitary
transformations without changing the state. This gives additional
freedom in the coherence-erasure optimization and can change, or even
reduce, the corresponding distance. For example, the maximally mixed
qutrit $\rho=I/3$ is already incoherent and therefore satisfies
$D_{\mathcal I}^{U}(\rho)=0$.

A complete treatment of partially degenerate qutrits requires the more
general optimization in Eq.~\eqref{eq:qutrit_degenerate_reduction},
which is beyond the scope of the present work. We therefore do not make
any general statement here about the behavior of
$D_{\mathcal I}^{U}$ across spectral degeneracies.
\end{remark}

\section{Operational interpretation and examples}
\label{sec:operational_examples}

\subsection{Minimum time for coherence erasure under bounded driving}
\label{subsec:erasure_time}

The unitary-accessible coherence-erasure distance has a direct
dynamical interpretation in terms of the minimum time required to
eliminate coherence under bounded unitary driving. To make this
connection explicit, consider a closed-system evolution
\begin{equation*}
    \dot{\rho}_t
    =
    -i[H_t,\rho_t],
    \qquad
    \rho_{t=0}=\rho,
\end{equation*}
generated by a time-dependent Hamiltonian satisfying
\begin{equation}
    \|H_t\|_\infty
    \le
    h,
    \qquad
    h>0.
    \label{eq:hamiltonian_bound}
\end{equation}
We assume the same unrestricted unitary-control setting as in the
definition of the quantum charging distance, with the only
constraint being the bound in
Eq.~\eqref{eq:hamiltonian_bound}.

We define the minimum unitary coherence-erasure time by
\begin{equation*}
\begin{aligned}
    \tau_{\mathcal I}^{U}(\rho;h)
    :=
    \inf
    \Bigl\{
        \tau\ge0:
        \rho_\tau\in\mathcal I_\rho,\,
        \|H_t\|_\infty\le h
        \ \text{for }0\le t\le\tau
    \Bigr\}.
    \label{eq:minimum_erasure_time}
\end{aligned}
\end{equation*}
Since, for a fixed unitary-accessible incoherent target
$\delta\in\mathcal I_\rho$, the quantum charging distance
$D_{\rm CD}(\rho,\delta)$ gives the minimum evolution time for
unit driving strength, the corresponding minimum time under the
bound $\|H_t\|_\infty\le h$ is
\begin{equation*}
    \tau_{\rm CD}(\rho,\delta;h)
    =
    \frac{D_{\rm CD}(\rho,\delta)}{h}.
\end{equation*}
Minimizing over all unitary-accessible incoherent targets therefore
gives
\begin{equation}
    \tau_{\mathcal I}^{U}(\rho;h)
    =
    \frac{D_{\mathcal I}^{U}(\rho)}{h}.
\label{eq:erasure_time_distance}
\end{equation}
Eq.~\eqref{eq:erasure_time_distance} shows that
$D_{\mathcal I}^{U}$ is not only a geometric distance on the
unitary orbit. Under a fixed driving-strength constraint, it
directly determines the optimal time required to reach the
unitary-accessible incoherent set.

The distance hierarchy established in
Theorem~\ref{thm:geometric_hierarchy} then immediately leads to a
corresponding hierarchy of lower bounds on the coherence-erasure
time. Indeed,
\begin{equation*}
    \arcsin
    \sqrt{
        2T_{\rm orb}(\rho)
    }
    \le
    \Theta_{B}^{\rm orb}(\rho)
    \le
    D_{\mathcal I}^{U}(\rho),
\end{equation*}
and hence
\begin{equation}
    \frac{
        \arcsin\sqrt{2T_{\rm orb}(\rho)}
    }{h}
    \le
    \frac{
        \Theta_{B}^{\rm orb}(\rho)
    }{h}
    \le
    \tau_{\mathcal I}^{U}(\rho;h).
\label{eq:erasure_time_hierarchy}
\end{equation}
Thus, the orbit-restricted transport cost and the
orbit-restricted Bures angle provide geometric lower bounds on the
minimum coherence-erasure time, whereas
$D_{\mathcal I}^{U}$ gives the exact optimal time within the
present unitary-control model.

For pure states, the three geometric quantities coincide, and
Eq.~\eqref{eq:erasure_time_hierarchy} is saturated throughout:
\begin{equation}
    \tau_{\mathcal I}^{U}(\psi;h)
    =
    \frac{
        \Theta_{B}^{\rm orb}(\psi)
    }{h}
    =
    \frac{
        \arcsin\sqrt{2T_{\rm orb}(\psi)}
    }{h}.
    \label{eq:pure_state_erasure_time}
\end{equation}
For mixed states, by contrast, the inequalities can be strict,
reflecting the distinction between state-space proximity and the
unitary control required to rotate the eigenbasis into the
incoherent reference basis. The qubit example below illustrates
this separation explicitly.
\subsection{Qubit example}
\label{subsec:qubit_example}

To illustrate the distinction between the unitary-accessible
coherence-erasure distance and the corresponding state-space
distance, consider the one-parameter family of qubit states
\begin{equation}
    \rho_r
    =
    \frac12
    \left[
        I
        +
        r
        \left(
            \frac{\sqrt3}{2}\sigma_x
            +
            \frac12\sigma_z
        \right)
    \right],
    \qquad
    0<r\le1.
    \label{eq:qubit_example_family}
\end{equation}
The Bloch direction is fixed, while the Bloch radius $r$ controls
the mixedness of the state. In particular,
\begin{equation*}
    \frac{|r_z|}{r}
    =
    \frac12.
\end{equation*}
Therefore, by Theorem~\ref{thm:qubit_exact_distance},
\begin{equation*}
    D_{\mathcal I}^{U}(\rho_r)
    =
    \frac12
    \arccos\frac12
    =
    \frac{\pi}{6}.
    \label{eq:qubit_example_DI}
\end{equation*}
Thus, throughout the entire nondegenerate family, the
unitary-accessible coherence-erasure distance remains constant,
even though the purity of the state changes with $r$.

By contrast, the orbit-restricted Bures angle depends explicitly on
the Bloch radius. Using Proposition~\ref{prop:qubit_bures}, we obtain
\begin{equation*}
    \Theta_{B}^{\rm orb}(\rho_r)
    =
    \frac12
    \arccos
    \left(
        1-\frac{r^2}{2}
    \right).
    \label{eq:qubit_example_bures}
\end{equation*}
Hence the corresponding accessibility difference is
\begin{equation*}
    \Delta_{\rm acc}(\rho_r)
    =
    \frac{\pi}{6}
    -
    \frac12
    \arccos
    \left(
        1-\frac{r^2}{2}
    \right).
    \label{eq:qubit_example_gap}
\end{equation*}

Fig.~\ref{qubitfig1} illustrates these two distances as the purity
of the state is varied while the Bloch direction is kept fixed.
\begin{figure}[t]
    \centering
    \includegraphics[width=0.95\linewidth]{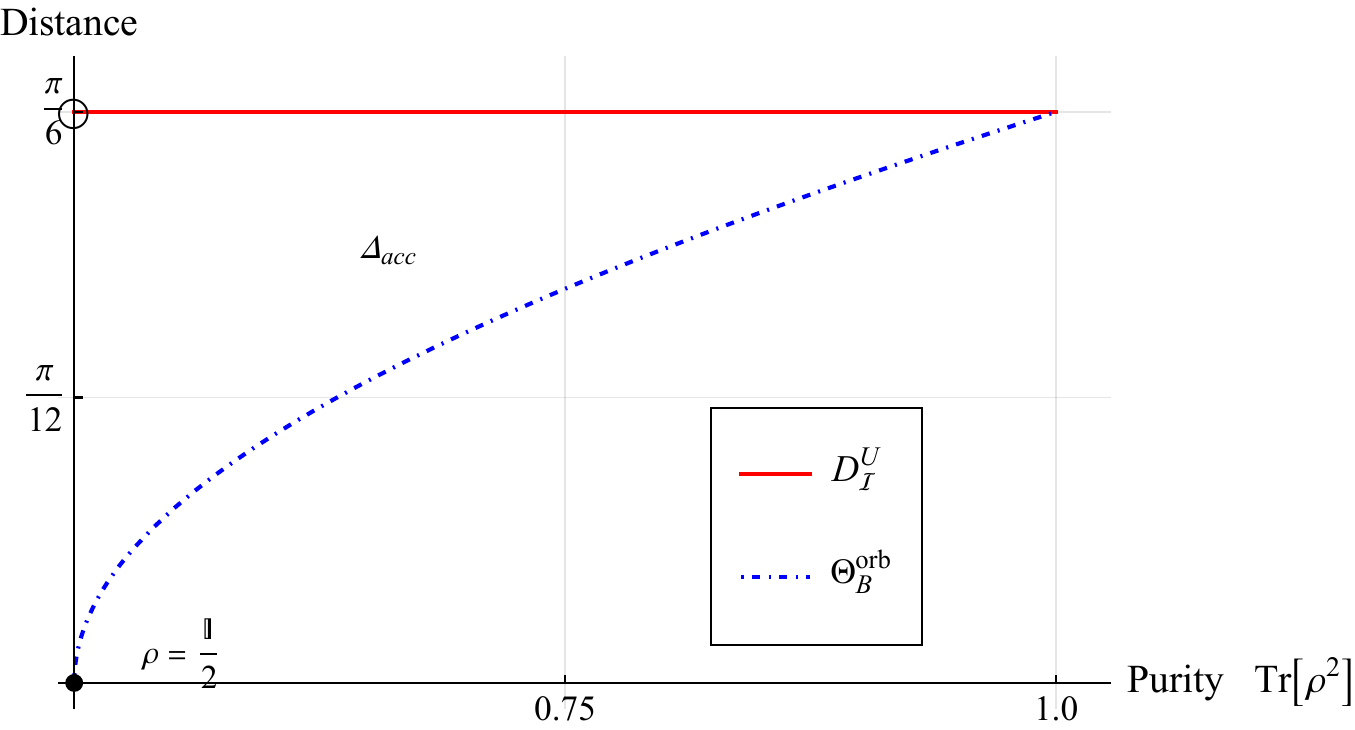}
    \caption{
    Unitary-accessible coherence-erasure distance
    $D_{\mathcal I}^{U}$ and orbit-restricted Bures angle
    $\Theta_B^{\rm orb}$ for the qubit family
    in Eq.~\eqref{eq:qubit_example_family}, with fixed Bloch direction
    $|r_z|/r=1/2$ and varying purity
    $\mathcal P=\operatorname{Tr}\rho^2$.
    In the nondegenerate sector,
    $D_{\mathcal I}^{U}=\pi/6$ is independent of the purity, whereas
    $\Theta_B^{\rm orb}$ decreases with increasing mixedness.
    Their separation is the accessibility difference
    $\Delta_{\rm acc}$.
    The open circle at $\mathcal P=1/2$ denotes the nondegenerate limit,
    while the filled point corresponds to the maximally mixed state
    $\rho=\mathbb I/2$, for which both distances vanish.
    }
    \label{qubitfig1}
\end{figure}
The two limiting regimes shown in Fig.~\ref{qubitfig1} are
particularly instructive. For the pure state endpoint $r=1$,
\begin{equation*}
    D_{\mathcal I}^{U}(\rho_1)
    =
    \Theta_{B}^{\rm orb}(\rho_1)
    =
    \frac{\pi}{6},
\end{equation*}
in agreement with the pure state equality discussed in
Sec.~\ref{sec:framework}.

As the state becomes increasingly mixed,
\begin{equation*}
    \Theta_{B}^{\rm orb}(\rho_r)
    \longrightarrow 0,
    \qquad
    r\to0^+,
\end{equation*}
whereas
\begin{equation*}
    D_{\mathcal I}^{U}(\rho_r)
    =
    \frac{\pi}{6}
\end{equation*}
remains unchanged throughout the nondegenerate sector.

At the maximally mixed state itself, however,
$\rho=\mathbb I/2$, both distances vanish. The open and filled
points in Fig.~\ref{qubitfig1} therefore distinguish the
nondegenerate limit $r\to0^+$ from the fully degenerate state
$r=0$.

This example makes the geometric distinction particularly
transparent. As the state becomes more mixed, the Bures angle
decreases and the accessibility difference grows, whereas the
unitary-accessible coherence-erasure distance remains fixed
throughout the nondegenerate family. The latter continues to measure
the finite rotation of the eigenbasis relative to the incoherent
reference basis.

\subsection{Qutrit examples and saturation of the bound}
\label{subsec:qutrit_examples}

We now illustrate the qutrit result by fixing the spectrum and
varying only the eigenbasis. Consider the nondegenerate spectrum
\begin{equation}
    \Lambda
    =
    \operatorname{diag}
    \left(
        \frac35,
        \frac3{10},
        \frac1{10}
    \right),
    \label{eq:qutrit_example_spectrum}
\end{equation}
and the family
\begin{equation*}
    \rho(R)
    =
    R\Lambda R^\dagger,
    \qquad
    R\in U(3).
    \label{eq:qutrit_example_family}
\end{equation*}
Since the spectrum is fixed and nondegenerate,
$D_{\mathcal I}^{U}(\rho(R))$ depends only on the orientation of
the eigenbasis $R$ relative to the reference basis.

We consider three representative choices of $R$.
First, for $R_0=I_3$, we have $\rho_0=\Lambda$, which is already
incoherent. Hence,
\begin{equation*}
    D_{\mathcal I}^{U}(\rho_0)=0.
    \label{eq:qutrit_example_identity}
\end{equation*}

Next, consider the qutrit Fourier basis
\begin{equation*}
    F_3
    =
    \frac1{\sqrt3}
    \begin{pmatrix}
        1 & 1 & 1\\
        1 & \omega & \omega^2\\
        1 & \omega^2 & \omega
    \end{pmatrix},
    \qquad
    \omega=e^{2\pi i/3}.
    \label{eq:qutrit_Fourier_basis}
\end{equation*}
For the same spectrum,
\begin{equation*}
    \rho_F
    :=
    F_3\Lambda F_3^\dagger
    =
    \begin{pmatrix}
        \dfrac13
        &
        \dfrac{2}{15}-\dfrac{i\sqrt3}{30}
        &
        \dfrac{2}{15}+\dfrac{i\sqrt3}{30}
        \\[2mm]
        \dfrac{2}{15}+\dfrac{i\sqrt3}{30}
        &
        \dfrac13
        &
        \dfrac{2}{15}-\dfrac{i\sqrt3}{30}
        \\[2mm]
        \dfrac{2}{15}-\dfrac{i\sqrt3}{30}
        &
        \dfrac{2}{15}+\dfrac{i\sqrt3}{30}
        &
        \dfrac13
    \end{pmatrix}.
    \label{eq:qutrit_example_rhoF}
\end{equation*}
Since $\mu(F_3)=1/2$,
\begin{equation*}
    D_{\mathcal I}^{U}(\rho_F)
    =
    \arccos\frac12
    =
    \frac{\pi}{3}
    \simeq
    60^\circ.
    \label{eq:qutrit_example_Fourier}
\end{equation*}
Thus, although the Fourier basis is mutually unbiased with respect
to the reference basis, it does not attain the qutrit bound.

Finally, consider the extremal eigenbasis $R_\ast$ introduced in
Sec.~\ref{subsec:qutrit_extremal}. For the same spectrum,
\begin{equation*}
    \rho_\ast
    :=
    R_\ast\Lambda R_\ast^\dagger
    =
    \begin{pmatrix}
        \dfrac7{20}
        &
        -\dfrac{\sqrt2}{8}
        &
        -\dfrac{\sqrt2}{8}
        \\[2mm]
        -\dfrac{\sqrt2}{8}
        &
        \dfrac{13}{40}
        &
        \dfrac1{40}
        \\[2mm]
        -\dfrac{\sqrt2}{8}
        &
        \dfrac1{40}
        &
        \dfrac{13}{40}
    \end{pmatrix}.
    \label{eq:qutrit_example_rhostar}
\end{equation*}
By Proposition~\ref{prop:qutrit_extremal},
\begin{equation*}
    D_{\mathcal I}^{U}(\rho_\ast)
    =
    \theta_\ast
    =
    \arccos
    \left(
        \frac{2\sqrt2-1}{4}
    \right)
    \simeq
    62.80^\circ.
    \label{eq:qutrit_example_extremal}
\end{equation*}
The three states have the same spectrum but different eigenbases,
and their distances satisfy
\begin{equation*}
    D_{\mathcal I}^{U}(\rho_0)
    =
    0,
    \qquad
    D_{\mathcal I}^{U}(\rho_F)
    =
    \frac{\pi}{3},
    \qquad
    D_{\mathcal I}^{U}(\rho_\ast)
    =
    \theta_\ast.
    \label{eq:qutrit_three_examples}
\end{equation*}
Thus, changing only the eigenbasis can vary the distance from zero
to the maximal qutrit value.

To show where these special cases lie among more general qutrit
eigenbases, we also sample $200$ matrices independently from the Haar
measure on $U(3)$. For each sampled matrix $R$, we compute
\begin{equation*}
    D_{\mathcal M}(R)
    :=
    \min_{M\in\mathcal M_3}
    \left\|
        i\operatorname{ln}(MR^\dagger)
    \right\|_\infty.
    \label{eq:qutrit_sample_distance}
\end{equation*}
For the fixed nondegenerate spectrum in
Eq.~\eqref{eq:qutrit_example_spectrum},
\begin{equation*}
    D_{\mathcal M}(R)
    =
    D_{\mathcal I}^{U}\bigl(\rho(R)\bigr).
\end{equation*}
Fig.~\ref{fig:qutrit_random_sampling} shows the normalized
histogram of the sampled values. The Fourier value $\pi/3$ and the
tight bound $\theta_\ast$ are marked by the blue dashed and red solid
lines, respectively. The sampling is included only to illustrate
the location of the special eigenbases relative to typical qutrit
eigenbases; the bound itself is established analytically.
\begin{figure}[t]
    \centering
    \includegraphics[width=0.95\linewidth]{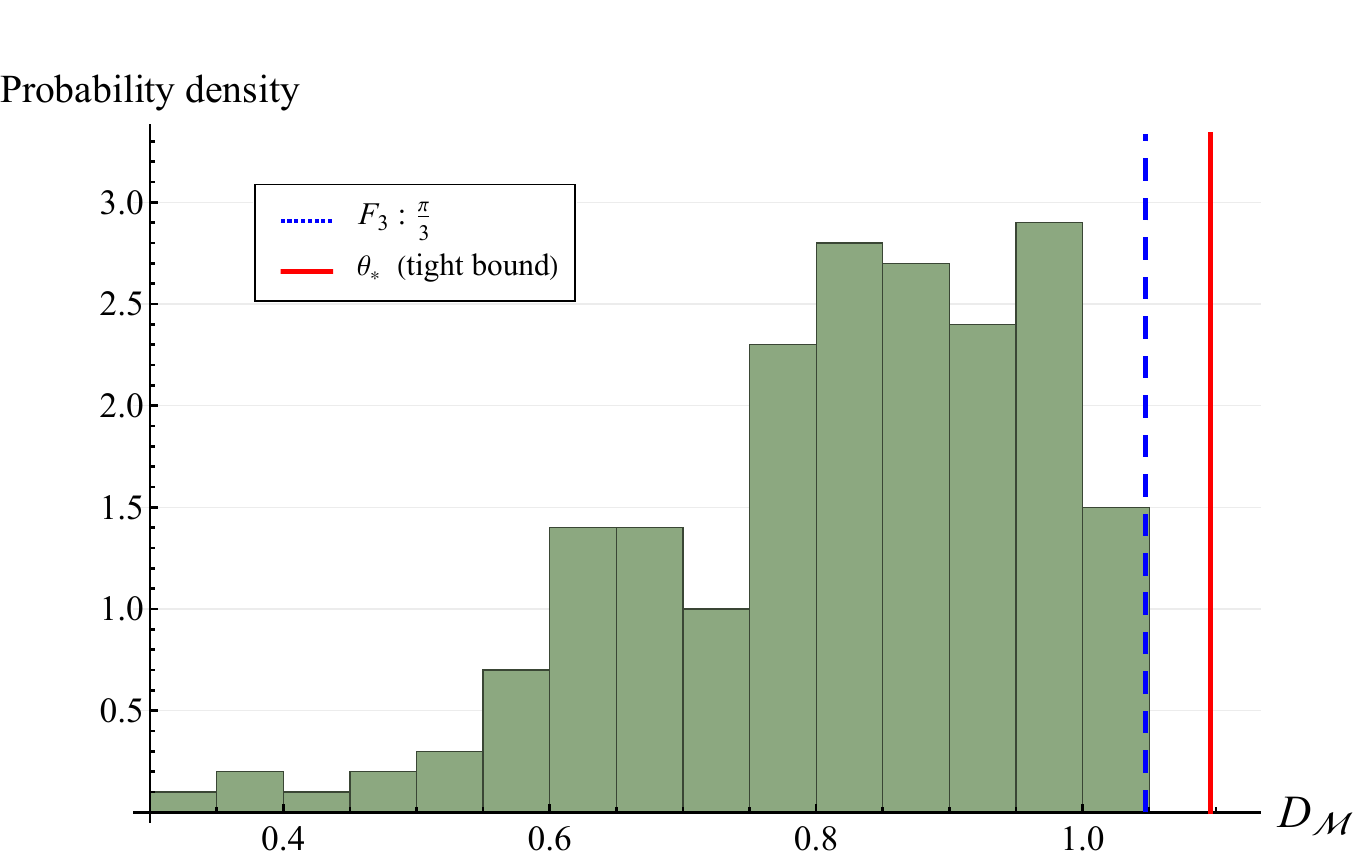}
    \caption{
    Normalized histogram of $D_{\mathcal M}(R)$ for
    $200$ Haar-random qutrit eigenbases $R\in U(3)$.
    The blue dashed line denotes the Fourier-basis value
    $D_{\mathcal M}(F_3)=\pi/3$, while the red solid line marks the
    tight qutrit bound
    $\theta_\ast=\arccos[(2\sqrt2-1)/4]$.
    The analytically constructed eigenbasis $R_\ast$ attains the
    bound.
    }
    \label{fig:qutrit_random_sampling}
\end{figure}
Under the bounded driving considered in
Sec.~\ref{subsec:erasure_time}, the same comparison applies to the
minimum erasure time. If $\|H_t\|_\infty\le h$, then
\begin{equation*}
    \tau_{\mathcal I}^{U}(\rho_0;h)=0,
    \qquad
    \tau_{\mathcal I}^{U}(\rho_F;h)
    =
    \frac{\pi}{3h},
    \qquad
    \tau_{\mathcal I}^{U}(\rho_\ast;h)
    =
    \frac{\theta_\ast}{h}.
    \label{eq:qutrit_three_erasure_times}
\end{equation*}
Together with the qubit example in
Sec.~\ref{subsec:qubit_example}, these results illustrate the role of
the eigenbasis and the spectrum from two different directions. In
the qubit example, the eigenbasis is fixed while the spectrum varies,
and $D_{\mathcal I}^{U}$ remains unchanged as long as the state is
nondegenerate. Here, the spectrum is fixed while the eigenbasis
varies, and the distance ranges from zero to the tight qutrit bound
$\theta_\ast$. Thus, for nondegenerate states, the unitary coherence
erasure distance is determined by the eigenbasis rather than by the
particular eigenvalues.
\section{Conclusions}
\label{sec:conclusion}

In this work, we introduced the unitary-accessible coherence-erasure
distance to quantify the minimum unitary-control cost required to
transform a quantum state into an incoherent state while remaining on
its unitary orbit. We established the ordering relation among the
orbit-restricted transport cost, the orbit-restricted Bures angle, and
the unitary-accessible distance, as summarized in
Theorem~\ref{thm:geometric_hierarchy}. The bounds become equal for pure
states, whereas mixed states can exhibit a strict separation between
state-space proximity and the unitary rotation required for coherence
erasure.

For qubits, we obtained an exact state-dependent expression for
$D_{\mathcal I}^{U}$ in Theorem~\ref{thm:qubit_exact_distance}. The
result shows that, for all nondegenerate qubit states, the
coherence-erasure cost depends on the orientation of the eigenbasis
relative to the incoherent reference basis rather than on the
eigenvalues themselves. The comparison with the orbit-restricted Bures
angle further demonstrates that this unitary-control distance can
remain finite even when the corresponding state-space distance becomes
arbitrarily small.

For nondegenerate qutrits, the problem was reduced to a
monomial unitary optimization and a tight universal bound was derived
in Eq.~\eqref{eq:qutrit_tight_bound}. An explicit eigenbasis
$R_\ast$ was shown to saturate this bound, establishing the exact
worst-case qutrit coherence-erasure distance. Together with the
minimum-time relation in Eq.~\eqref{eq:erasure_time_distance}, these
results provide a direct dynamical interpretation of the geometric
distance under bounded unitary driving.

Several natural directions remain for future study. In higher dimensions, it would be
natural to determine the corresponding monomial unitary covering
radius and characterize its saturating configurations. Spectral
degeneracies require an enlarged optimization over the stabilizer of
the spectrum, while physically restricted Hamiltonians may modify the
attainable geodesics and hence the minimum erasure time. Extensions
beyond closed unitary dynamics also provide a natural direction for
future work.

\section*{ACKNOWLEDGEMENTS}
This work is supported by the National Natural Science Foundation of China (NSFC) under Grant No. 12171044, the Natural Science Foundation of Jilin Province, China (Grant No. 20250102018JC), and the Specific Research Fund of the Innovation Platform for Academicians of Hainan Province.


\appendix

\section{Proof of the tight qutrit bound}
\label{app:qutrit_bound}

In this appendix we prove the tight qutrit bound stated in
Theorem~\ref{thm:qutrit_tight_bound}. For a nondegenerate qutrit
state $\rho=R\Lambda R^\dagger$, $R\in U(3)$,
$D_{\mathcal I}^{U}(\rho)=\arccos\mu(R)$,
where
\begin{equation}
\mu(R)
:=
\max_{M\in\mathcal M_3}
\lambda_{\min}
\left[
\operatorname{Re}(MR^\dagger)
\right].
\label{eq:app_mu}
\end{equation}
It is therefore sufficient to prove
\[
\mu(R)\ge c_\ast:=\frac{2\sqrt2-1}{4},
\qquad
R\in U(3),
\]
and to identify a matrix attaining equality. We also set
$\theta_\ast:=\arccos c_\ast$.

Define the maximal matching sum
\begin{equation}
S(R)
:=
\max_{\pi\in S_3}
\sum_{i=1}^{3}|R_{i,\pi(i)}|,
\label{eq:app_matching}
\end{equation}
and let $S_\ast:=\sqrt2+\frac12$, so that
$c_\ast=(S_\ast-1)/2$.

We use the following invariance property.

\begin{lemma}
\label{lem:app_monomial_invariance}
Let $A,B\in\mathcal M_3$. Then $\mu(ARB)=\mu(R)$ and
$\mu(R^T)=\mu(R)$. In particular, row and column permutations,
diagonal rephasings, and transposition do not change $\mu(R)$.
\end{lemma}

\begin{proof}
Let $R'=ARB$. For any $M\in\mathcal M_3$, define
$N=A^\dagger M B^\dagger$. Since $\mathcal M_3$ is a group,
$N\in\mathcal M_3$, and $M\mapsto N$ is bijective. Moreover,
$MR'^\dagger=ANR^\dagger A^\dagger$, and hence
\[
\operatorname{Re}(MR'^\dagger)
=
A\,\operatorname{Re}(NR^\dagger)A^\dagger.
\]
Thus the two Hermitian matrices have the same minimum eigenvalue.
Taking the maximum over $\mathcal M_3$ gives
$\mu(ARB)=\mu(R)$.

For the transpose, $(MR^\dagger)^T=\overline R M^T$.
Transposition preserves the eigenvalues of the Hermitian part, while
$\overline R M^T$ is unitarily similar to
$M^T\overline R=M^T(R^T)^\dagger$. Since
$M\mapsto M^T$ is a bijection of $\mathcal M_3$,
$\mu(R^T)=\mu(R)$.
\end{proof}

We first consider the large-matching region, then the interior and
orthostochastic boundary of the hard region.

\subsection{The large-matching region}
\label{app:large_matching}

\begin{lemma}
\label{lem:large_matching}
Let $R\in U(3)$. If $S(R)\ge S_\ast$, then
\[
\mu(R)
\ge
\frac{S(R)-1}{2}
\ge
c_\ast.
\]
\end{lemma}

\begin{proof}
Choose a permutation attaining the maximum in
Eq.~\eqref{eq:app_matching}. The diagonal phases in $\mathcal M_3$
allow us to choose $M_0\in\mathcal M_3$ such that
$\operatorname{tr}(M_0R^\dagger)=S(R)>0$.
Set $W=M_0R^\dagger$.

Let $2\delta$ be the length of the shortest closed arc containing the
three eigenvalues of $W$. Since $|\operatorname{tr}W|=S(R)>1$, these
eigenvalues lie in an open semicircle; otherwise three points on the
unit circle not contained in any open semicircle have a vector sum of
modulus at most one. Hence $0\le\delta<\pi/2$.

Choose $\alpha$ so that the shortest containing arc of
$e^{i\alpha}W$ is centered at zero. Its extreme eigenvalues are
$e^{-i\delta}$ and $e^{i\delta}$, while the third is $e^{ix}$ with
$|x|\le\delta$. Therefore
$
S(R)
=
|\operatorname{tr}W|
=
\left|2\cos\delta+e^{ix}\right|
\le
1+2\cos\delta,
$ and hence
\begin{equation}
\cos\delta
\ge
\frac{S(R)-1}{2}.
\label{eq:app_delta_bound}
\end{equation}
Moreover,
$
\lambda_{\min}
\left[
\operatorname{Re}(e^{i\alpha}W)
\right]
=
\cos\delta.
$
Since $e^{i\alpha}M_0\in\mathcal M_3$,
Eq.~\eqref{eq:app_delta_bound} gives
\[
\mu(R)
\ge
\frac{S(R)-1}{2}
\ge
\frac{S_\ast-1}{2}
=
c_\ast.
\]
\end{proof}

\begin{lemma}
\label{lem:matching_lower_bound}
For every $R\in U(3)$, $S(R)\ge\sqrt3$.
\end{lemma}

\begin{proof}
Let $P_{ij}:=|R_{ij}|^2$. Since $R$ is unitary, $P$ is doubly
stochastic. The van der Waerden bound gives
$\operatorname{per}P\ge2/9$. Since the permanent contains six
nonnegative terms, some $\pi\in S_3$ satisfies
\[
\prod_{i=1}^{3}P_{i,\pi(i)}
\ge\frac1{27},
\qquad
\prod_{i=1}^{3}|R_{i,\pi(i)}|
\ge\frac1{3\sqrt3}.
\]
AM--GM then gives
\[
\sum_{i=1}^{3}|R_{i,\pi(i)}|\ge\sqrt3.
\]
Taking the maximum over $\pi$ proves the claim.
\end{proof}

Thus it remains to consider the hard region
$\sqrt3\le S(R)<S_\ast$.

\subsection{The hard unistochastic interior}
\label{app:qutrit_interior}

By Lemma~\ref{lem:app_monomial_invariance}, we may choose a maximal
matching to be diagonal and use diagonal rephasings to make its
entries real and nonnegative. Write
\[
|R|
=
\begin{pmatrix}
a & u & t\\
v & b & w\\
r & z & c
\end{pmatrix},
\qquad
a,b,c,u,v,t,w,r,z\ge0,
\]
so that $S=a+b+c$.

Maximality of the diagonal matching gives
\begin{align}
T_{12}:&\quad u+v\le a+b,
\label{eq:app_T12}\\
T_{23}:&\quad w+z\le b+c,
\label{eq:app_T23}\\
T_{13}:&\quad t+r\le a+c,
\label{eq:app_T13}\\
C_{+}:&\quad u+w+r\le S,
\label{eq:app_Cplus}\\
C_{-}:&\quad t+v+z\le S.
\label{eq:app_Cminus}
\end{align}
A constraint is called active when equality holds. Define the
corresponding slacks
\begin{align*}
\Delta_{T_{12}}&:=a+b-u-v,\\
\Delta_{T_{23}}&:=b+c-w-z,\\
\Delta_{T_{13}}&:=a+c-t-r,\\
\Delta_{C_+}&:=S-u-w-r,\\
\Delta_{C_-}&:=S-t-v-z.
\end{align*}

Up to relabeling, transposition, and interchange of the two cyclic
orientations, every nonempty active set reduces to one of the cases
below.

Write $\det R=e^{i\gamma}$. The cofactor identity
$ab-R_{12}R_{21}=ce^{i\gamma}$ gives
$u^2v^2=a^2b^2+c^2-2abc\cos\gamma$. Hence
\begin{equation}
\Phi
:=
\cos\gamma
=
\frac{a^2b^2+c^2-u^2v^2}{2abc}.
\label{eq:app_Phi}
\end{equation}

We next relate $\Phi$ to the smallest-eigenvalue optimization.
For the chosen diagonal matching, set $W=R^\dagger$. Since the
diagonal entries of $R$ are nonnegative, $\operatorname{tr}W=S>1$.
After centering the shortest containing arc of the eigenvalues of
$W$, write them as $e^{-i\delta}, e^{ix}, e^{i\delta}, |x|\le\delta,$
and define $\kappa:=\cos\delta$. For this matching, $\kappa$ is the
optimized minimum eigenvalue after the remaining overall phase
rotation, so
\begin{equation}
\mu(R)\ge\kappa.
\label{eq:app_mu_ge_kappa}
\end{equation}

Since
$S=\left|2\kappa+e^{ix}\right|$,
we have
\begin{equation}
\cos x
=
\frac{S^2-4\kappa^2-1}{4\kappa}.
\label{eq:app_cosx}
\end{equation}
If $\alpha$ is the centering phase, the trace and determinant give
$e^{-i\alpha}(2\kappa+e^{ix})=S$ and
$x\equiv3\alpha-\gamma\pmod{2\pi}$. Eliminating $\alpha$ and $x$
yields
\begin{widetext}
\begin{equation}
\Phi
=
\Psi(S,\kappa)
:=
\frac{
S^4+16S^2\kappa^4+4S^2\kappa^2-2S^2
-64\kappa^6+48\kappa^4-12\kappa^2+1
}{
8S^3\kappa^2
}.
\label{eq:app_Phi_kappa}
\end{equation}
\end{widetext}
Differentiation gives
\begin{equation}
\frac{\partial\Psi}{\partial\kappa}
=
-
\frac{
(4\kappa^2-1-S)
(4\kappa^2-1+S)
(8\kappa^2+1-S^2)
}{
4S^3\kappa^3
}.
\label{eq:app_dPhi_kappa}
\end{equation}
Because $|x|\le\delta$,
\[
8\kappa^2+1-S^2\le0,
\qquad
\kappa\ge\frac{S-1}{2}.
\]
Moreover, $\kappa<1$ since $S<3$. Using
$S^2\ge8\kappa^2+1$, we obtain
$S>4\kappa^2-1$, while
$4\kappa^2-1+S>S-1>0$. Hence
$\partial_\kappa\Psi<0$ in the interior of the admissible branch.
Thus minimizing $\kappa$ at fixed $S$ is equivalent to maximizing
$\Phi$.

For fixed $\sqrt3<S<S_\ast$, the squared-modulus matrices satisfying
normalization and the five matching inequalities form a compact
unistochastic feasible set. By Lemma~\ref{lem:orthostochastic} below,
its boundary satisfies $S\ge S_\ast$ and is therefore excluded.
Thus $\Phi$ attains a maximum in the interior.

At $S=\sqrt3$, every matching sum is at most $\sqrt3$. AM--GM gives
each permanent term at most $1/27$, so equality in the permanent
bound forces all moduli to equal $1/\sqrt3$. Then
Eqs.~\eqref{eq:app_Phi}--\eqref{eq:app_Phi_kappa} give
$\kappa=1/2$.

For $\sqrt3<S<S_\ast$, use $a,b,u^2$ as local coordinates and set
$x=u^2$, $c=S-a-b$. Normalization gives
\[
t^2=1-a^2-x,\qquad
z^2=1-b^2-x,\qquad
v^2=1-a^2-b^2-x+c^2.
\]

If no matching constraint is active, then
\[
\Phi(a,b,x)
=
\frac{
a^2b^2+c^2-x(1-a^2-b^2-x+c^2)
}{
2abc
}.
\]
At fixed $S,a,b$,
$\partial_x^2\Phi=1/(abc)>0$, so no interior maximum occurs.

Suppose next that $C_+$ is active. Normalization gives
$w^2-a^2=u^2-c^2$ and $r^2-b^2=u^2-c^2$, hence
\begin{equation}
u+w+r-S
=
(u-c)
\left[
1+\frac{u+c}{w+a}
+\frac{u+c}{r+b}
\right].
\label{eq:app_cycle_factorization}
\end{equation}
The bracket is positive, so
\begin{equation}
u=c,\qquad w=a,\qquad r=b.
\label{eq:app_cycle_rigidity}
\end{equation}
Thus
$
\Phi_C
=
\frac12
\left(
\frac{ab}{c}
+
\frac{ac}{b}
+
\frac{bc}{a}
\right)$,
$
a+b+c=S.
$
Its only interior stationary point is $a=b=c=S/3$. The Hessian in
two independent coordinates is
\[
\frac3a
\begin{pmatrix}
1&1/2\\
1/2&1
\end{pmatrix},
\]
which is positive definite. Hence this point is a strict minimum, so
a maximum on the single-cycle stratum requires another active
constraint.

Suppose a single transposition is active. Up to relabeling, take
$u+v=a+b$ and set
$p:=a+b$, $q:=a-b$, $c=S-p$, and $d:=u-v$. Then
\begin{equation}
d^2
=
2+2c^2-2p^2-q^2.
\label{eq:app_d_oneT}
\end{equation}
Substitution into Eq.~\eqref{eq:app_Phi} gives
\[
\Phi_T(p,q)
=
\frac{
(p^2-q^2)^2-(p^2-d^2)^2+16c^2
}{
8c(p^2-q^2)
}.
\]
At fixed $S$ and $p$,
\[
\frac{\partial\Phi_T}{\partial q}
=
-
\frac{
q[2p^2-(1+c)^2][2p^2-(1-c)^2]
}{
c(p^2-q^2)^2
}.
\]
Thus an interior stationary point must satisfy
$q=0$, $2p^2=(1+c)^2$, or $2p^2=(1-c)^2$.

On $q=0$,
\[
\left.
\frac{\partial\Phi_T}{\partial p}
\right|_{q=0}
=
-
\frac{
(S^2-1)P_S(p)
}{
2p^3(S-p)^2
},
\]
where
$
P_S(p)
=
5p^3-8Sp^2+(7S^2-3)p-2S^3+2S.
$
Since
$P_S'(p)=15p^2-16Sp+7S^2-3$ has discriminant
$180-164S^2<0$ for $S\ge\sqrt3$, $P_S$ is strictly increasing.
Moreover, $P_S(S-1)=2(S-1)^3>0$. Hence a zero would require
$p<S-1$, so $c=S-p>1$, which is impossible.

The branch $2p^2=(1-c)^2$ gives
$S=c+(1-c)/\sqrt2\le1$, contradicting $S\ge\sqrt3$.
For $2p^2=(1+c)^2$, $p=(1+c)/\sqrt2$ and $c>1/2$. At these points,
\begin{widetext}
\[
\left.
\frac{\partial\Phi_T}{\partial p}
\right|_{2p^2=(1+c)^2}
=
\frac{
4(1+\sqrt2)(1+c)
(c^2-6c-2q^2+1)
(c^2+2c-2q^2+1)
}{
c(\sqrt2 c-2q+\sqrt2)^2
(\sqrt2 c+2q+\sqrt2)^2
}.
\]
\end{widetext}
The second factor is
$c^2+2c-2q^2+1=2(p^2-q^2)>0$, while the first is negative for
$c>1/2$. Hence no genuine single-transposition branch contains an
interior maximum.

At the single-$T$ fold, write $s=S\in(\sqrt3,S_\ast)$.
Equation~\eqref{eq:app_d_oneT} gives local coordinates $(q,d)$.
Fixing $q=q_0$,
\[
p(d)=p_0-\frac{d^2}{4s},
\qquad
c(d)=c_0+\frac{d^2}{4s},
\]
where $c_0=s-p_0$. Set
$\Delta_0:=p_0^2-q_0^2>0$ and $\Phi_0:=\Phi(0)$. Expanding gives
\[
\Phi(d)=\Phi_0+K_Td^2+O(d^4),
\]
with
$
K_T=
\frac{
p_0q_0^2+2sp_0^2+8c_0
-2\Phi_0(\Delta_0-2p_0c_0)
}{
8sc_0\Delta_0
}.
$
Since $|\Phi_0|\le1$, if $\Delta_0\ge2p_0c_0$ the numerator is at
least
$
(p_0+2)q_0^2
+2(s-1)p_0^2
+8c_0
+4p_0c_0
>0.
$
If $\Delta_0<2p_0c_0$, it is at least
$
p_0q_0^2
+2sp_0^2
+2\Delta_0
+4c_0(2-p_0)
>0,
$
since $p_0<2$. Thus $K_T>0$.

For a $TC$ configuration, take $T_{12}$ and $C_+$ active.
Equation~\eqref{eq:app_cycle_factorization} gives
$u=c, w=a, r=b, v=a+b-c,$
and hence
$
a+b=\frac{S+v}{2},
c=\frac{S-v}{2}.
$
Writing $q=a-b$, normalization yields
\begin{equation}
q^2
=
\frac{8-S^2-2Sv-9v^2}{4}.
\label{eq:app_TC_q}
\end{equation}
The reduced phase function is
\begin{widetext}
\[
\Phi_{TC}(S,v)
=
\frac{
S^4+4S^3v-2S^2v^2+8S^2
+52Sv^3-48Sv+9v^4-24v^2+16
}{
8(S-v)(S^2+2Sv+5v^2-4)
}.
\]
\end{widetext}
Differentiation gives
\[
\frac{\partial\Phi_{TC}}{\partial v}
=
-
\frac{
A(S,v)F(S,v)
}{
8(S-v)^2(S^2+2Sv+5v^2-4)^2
},
\]
where
$
A(S,v)=3S^2+2Sv+3v^2-4
$
and
\[
F(S,v)
=
-S^4+4S^3v-54S^2v^2+24S^2
-28Sv^3+16Sv+15v^4+24v^2-16.
\]
Clearly $A(S,v)>0$. Writing $S=2c+v$,
\[
\frac{F(S,v)}{16}
=
D_0Q
+
4\left[
2cv^3+(3v^2-1)(2v^2-1)
\right],
\]
where
\[
D_0=c^2+2cv+2v^2-1=2ab>0,
\qquad
Q=-c^2+2cv+5-14v^2.
\]
From $q^2\ge0$ and $2c+v\ge\sqrt3$,
$0\le v\le1/\sqrt3$, so
$(3v^2-1)(2v^2-1)\ge0$. Moreover,
\[
Q
\ge
3+4cv-11v^2
\ge
3+2\sqrt3\,v-13v^2
\ge
\frac23.
\]
Thus $F(S,v)>0$ and
\begin{equation}
\frac{\partial\Phi_{TC}}{\partial v}<0.
\label{eq:app_TC_monotonic}
\end{equation}

At the $TC$ fold $q=0$, Eq.~\eqref{eq:app_TC_q} makes $q$ a local
coordinate. If $v_0=v(0)>0$, then
\[
v(q)
=
v_0-\frac{2q^2}{s+9v_0}+O(q^4),
\qquad
\Phi(q)
=
\Phi(0)+K_{TC}q^2+O(q^4),
\]
where
\[
K_{TC}
=
-\frac{2\,\partial_v\Phi_{TC}(s,v_0)}{s+9v_0}>0.
\]
Hence no $TC$ fold is a local maximum. Moreover,
$
\Delta_{C_-}
=
\Delta_{T_{13}}
+
\Delta_{T_{23}},
$
so $C_-$ cannot become active while both $T_{13}$ and $T_{23}$ remain
inactive. The relevant enlargement is therefore a $TTC$ configuration.

Next let $T_{12}$ and $T_{23}$ be active. Normalization gives
\[
w^2-u^2=a^2-c^2,
\qquad
z^2-v^2=a^2-c^2,
\]
and hence
\[
w+z-(b+c)
=
(a-c)
\left[
1+\frac{a+c}{w+u}
+\frac{a+c}{z+v}
\right].
\]
Since $T_{23}$ is active and the bracket is positive,
\begin{equation}
a=c.
\label{eq:app_TT_ac}
\end{equation}
Thus $w=u$, $z=v$, and $S=2a+b$. The reduced phase function is
\begin{widetext}
\[
\Phi_{TT}(a,b)
=
-
\frac{
a^4+4a^3b+4a^2b^2-6a^2
+8ab^3-4ab+4b^4-4b^2+1
}{
8a^2b
}.
\]
\end{widetext}

At fixed $S$, write
\[
a=\frac{S-b}{2},
\qquad
\varrho:=\frac{a}{b},
\qquad
y:=b^2.
\]
Normalization and $u+v=a+b$ give
\[
(u-v)^2=2-2b^2-(a+b)^2\ge0.
\]
Its right-hand side is decreasing in $b$, and at $b=S/3$ equals
$2-2S^2/3\le0$. Hence $b\le S/3$ and $a\ge b$.

One obtains
\[
\frac{d\Phi_{TT}}{db}
=
-
\frac{m(\varrho,y)}{8a^3b},
\]
where
$
m(\varrho,y)
=
A(\varrho)y^2+B(\varrho)y+1-\varrho,
$
with
$
A(\varrho)
=
-(\varrho^2-2\varrho-2)
(\varrho^3+3\varrho^2+6\varrho+2),
$
and
$
B(\varrho)
=
2(3\varrho^3-3\varrho-2).
$

The remaining cycle inequality gives $a^2+2b^2\ge1$. Together with
$(u-v)^2\ge0$,
\[
(\varrho^2+2)y\ge1,
\qquad
(\varrho^2+2\varrho+3)y\le2,
\qquad
S^2
\ge
\frac{(2\varrho+1)^2}{\varrho^2+2}.
\]
The last lower bound is increasing on $[1,4]$ and equals
$S_\ast^2$ at $\varrho=\sqrt2$; for $\varrho\ge4$ it is at least
$4>S_\ast^2$. Hence $S<S_\ast$ implies
$\varrho<\sqrt2$. We may therefore use
\[
1\le\varrho\le\sqrt2,
\qquad
\frac14\le y\le\frac13.
\]
On this range $A(\varrho)>0$, and
\[
\frac{\partial m}{\partial y}
\ge
\frac{A(\varrho)}2+B(\varrho)
=
-\frac12
\left(
\varrho^5+\varrho^4-14\varrho^3
-16\varrho^2-4\varrho+4
\right)>0.
\]
Hence
\[
m(\varrho,y)
\ge
m\left(\varrho,\frac14\right),
\]
while
\[
m\left(\varrho,\frac14\right)-\frac54
=
-
\frac{
(\varrho-1)
(\varrho^2-4\varrho-4)
(\varrho^2+6\varrho+4)
}{16}
\ge0.
\]
Therefore
\[
\frac{d\Phi_{TT}}{db}<0.
\]

At the $TT$ fold $d=u-v=0$,
\[
d^2
=
2-\frac{s^2}{4}-\frac{sb}{2}-\frac{9b^2}{4},
\]
so $d$ is a local coordinate. If $b_0=b(0)>0$, then
\[
b(d)
=
b_0-\frac{2d^2}{s+9b_0}+O(d^4),
\qquad
\Phi(d)
=
\Phi(0)+K_{TT}d^2+O(d^4),
\]
with
\[
K_{TT}
=
-\frac{2(d\Phi_{TT}/db)_{b=b_0}}{s+9b_0}>0.
\]

For the single-$T$, $TC$, and $TT$ folds above, these perturbations
preserve normalization, fixed $S$, and the indicated active
equalities. Since $S<S_\ast$ excludes the unistochastic boundary,
small perturbations remain feasible and inactive matching
inequalities remain strict. Hence these fold points are not local
maxima unless the active set enlarges.

Since $\Phi_{TT}$ decreases with $b$, its maximum occurs at the
smallest admissible $b$, where $a^2+2b^2=1$. At this endpoint the
remaining cycle constraint is saturated, and the $TT$ branch reaches
the $TTC$ critical ridge.

If both cyclic constraints are active, cycle rigidity in both
orientations gives
\[
u=v=c,\qquad
w=z=a,\qquad
r=t=b.
\]
The three transposition inequalities then force
$a=b=c=1/\sqrt3$, yielding the Fourier-type modulus matrix.
Likewise, if all three transpositions are active, the preceding
$TT$ argument applied cyclically again gives the Fourier-type
configuration. For this configuration,
$\kappa=1/2$ and hence $\mu(R)\ge1/2>c_\ast$.

Every remaining active set with at least three constraints contains
two transpositions and one cycle and is therefore of $TTC$ type.
Thus, up to the symmetries above, every fixed-$S$ interior maximizer
of $\Phi$ in the hard region lies on the $TTC$ critical ridge.

\subsection{The $TTC$ critical ridge}
\label{app:TTC_ridge}

Up to row and column permutations and transposition, the modulus
matrix on the $TTC$ ridge has the form
\[
|R|
=
\begin{pmatrix}
a & a & t\\
b & b & a\\
b & b & a
\end{pmatrix}.
\]
Normalization gives
\begin{equation}
a^2+2b^2=1,
\qquad
t^2=4b^2-1,
\label{eq:app_TTC_norm}
\end{equation}
while
\begin{equation}
S=2a+b.
\label{eq:app_TTC_S}
\end{equation}
Eliminating $a$ yields
$9b^2-2Sb+S^2-4=0$. Since $t^2\ge0$, $b\ge1/2$, so the admissible
branch on $\sqrt3\le S\le S_\ast$ is
\begin{equation}
b(S)
=
\frac{
S+2\sqrt{9-2S^2}
}{9},
\label{eq:app_bS}
\end{equation}
with $a(S)=(S-b(S))/2$.

Substituting the $TTC$ relations into Eq.~\eqref{eq:app_Phi} gives
\[
\Phi_{\rm TTC}(S)=\frac{1}{2b(S)}.
\]

\begin{lemma}
\label{lem:TTC_bound}
Along the $TTC$ critical ridge, the smallest-eigenvalue branch
associated with the maximal diagonal matching satisfies
$\kappa\ge c_\ast$. More precisely,
$\kappa>c_\ast$ for $\sqrt3\le S<S_\ast$, whereas
$\kappa=c_\ast$ at $S=S_\ast$.

Consequently, $\mu(R)>c_\ast$ for $\sqrt3\le S<S_\ast$, and
$\mu(R)\ge c_\ast$ at $S=S_\ast$.
\end{lemma}

\begin{proof}
Since $S\le S_\ast$ and
$8c_\ast^2+1<3\le S^2$, the value $c_\ast$ lies in the admissible
comparison range for $\kappa$. By
Eq.~\eqref{eq:app_dPhi_kappa},
\[
\kappa\ge c_\ast
\quad\Longleftrightarrow\quad
\Phi_{\rm TTC}(S)\le\Psi(S,c_\ast).
\]
Substituting Eq.~\eqref{eq:app_bS} and
$c_\ast=(2\sqrt2-1)/4$ reduces this inequality to $E(S)\ge0$, where
\begin{equation}
E(S)
:=
\frac{S_\ast-S}{72\,b(S)}
\left[
(4+3\sqrt2)(S_\ast+S)
+(2\sqrt2-1)\sqrt{9-2S^2}
\right].
\label{eq:app_ES}
\end{equation}

To verify that only positive factors were cleared, set
$\eta=\sqrt{9-2S^2}$,
\[
N
=
S^4+\frac{117-88\sqrt2}{16}S^2
+\frac{-605+428\sqrt2}{64},
\qquad
L=36S^2c_\ast^2-N,
\]
and let $\mathcal B(S)$ denote the bracket in
Eq.~\eqref{eq:app_ES}. Define
\begin{widetext}
\[
\begin{aligned}
Q_0(X)
=
X^4-\left(8\sqrt2-\frac{67}{8}\right)X^3
+\frac{12305-8656\sqrt2}{256}X^2
-\frac{569169-402460\sqrt2}{1024}X
+\frac{219073-154908\sqrt2}{256},
\end{aligned}
\]
\end{widetext}
$$\Psi(S,c_\ast)-\frac1{2b(S)}
=
\frac{
9(S_\ast+S)Q_0(S^2)
}{
S^3c_\ast^2(2\eta N+SL)\mathcal B(S)
}
E(S).
$$
This follows from
\[
(2\eta N-SL)(2\eta N+SL)
=
9(S_\ast^2-S^2)Q_0(S^2).
\]

The cubic coefficient of $Q_0$ is greater than $-3$, its quadratic
coefficient exceeds the magnitude of its linear coefficient, and its
constant is positive. Hence $Q_0(X)>0$ for $X\ge3$.
Writing $N=X^2+AX+B$, one has
$-1<A<0$ and $0<B<-3A$, so $0<N<X^2$ for $X\ge3$.
Since $36c_\ast^2>S_\ast^2\ge X$, also $L>0$.
Thus every factor in the multiplier above is positive on
$\sqrt3\le S\le S_\ast$.

Since $b(S)>0$ and the bracket in Eq.~\eqref{eq:app_ES} is positive,
the sign of $E(S)$ is determined by $S_\ast-S$. Hence
\[
E(S)>0
\quad
(\sqrt3\le S<S_\ast),
\qquad
E(S_\ast)=0.
\]
Because $\Psi(S,\kappa)$ is strictly decreasing in $\kappa$,
\[
\kappa>c_\ast
\quad
(\sqrt3\le S<S_\ast),
\qquad
\kappa=c_\ast
\quad
(S=S_\ast).
\]
The conclusions for $\mu(R)$ follow from
Eq.~\eqref{eq:app_mu_ge_kappa}.
\end{proof}

At the endpoints,
$
S=\sqrt3
\quad\Longrightarrow\quad
a=b=t=\frac1{\sqrt3},
$
corresponding to the Fourier-type configuration, while
$
S=S_\ast
\quad\Longrightarrow\quad
a=\frac1{\sqrt2},
\qquad
b=\frac12,
\qquad
t=0.
$
The latter will be shown below to attain the tight constant exactly.

\subsection{The orthostochastic boundary}
\label{app:orthostochastic}

If $P_{ij}=|R_{ij}|^2$ lies on the boundary of the
$3\times3$ unistochastic set, then $R$ is equivalent, up to left and
right diagonal rephasings, to a real orthogonal matrix. For positive
entries, the orthogonality triangle between two rows degenerates at
the boundary, so the rows can be rephased to be real. The third row
is then real up to an overall phase. If some entry vanishes, the same
conclusion follows by directly rephasing the remaining nonzero
entries. A row sign may finally be chosen so that the determinant is
positive.

By Lemma~\ref{lem:app_monomial_invariance}, it is therefore enough to
consider $O\in SO(3)$. Let $\mathcal C_3\subset SO(3)$ be the group of
orientation-preserving signed permutation matrices. Since
$\mathcal C_3\subset\mathcal M_3$, it is enough to optimize over this
subgroup.

\begin{lemma}
\label{lem:orthostochastic}
For every $O\in SO(3)$,
\[
\max_{C\in\mathcal C_3}
\lambda_{\min}
\left[
\operatorname{Re}(CO^T)
\right]
\ge
c_\ast.
\]
\end{lemma}

\begin{proof}
Represent $O$ by a unit quaternion
$q=(q_0,q_1,q_2,q_3)$,
$q_0^2+q_1^2+q_2^2+q_3^2=1$.
Modulo the binary octahedral group associated with $\mathcal C_3$,
every rotation may be represented in the Voronoi cell of the identity.
Taking $q_0>0$ and introducing Rodrigues coordinates
$x_i=q_i/q_0$, comparison with the nearest $\pi/2$ coordinate-axis
rotations gives
\[
|x_i|\le\alpha,
\qquad
\alpha:=\sqrt2-1,
\]
while comparison with the $2\pi/3$ body-diagonal rotations gives
\[
|x_1|+|x_2|+|x_3|\le1.
\]
Containment in this polytope is sufficient for the estimate below.

Let $r^2=x_1^2+x_2^2+x_3^2$. By symmetry take
$0\le x_3\le x_2\le x_1\le\alpha$. At a maximizer,
$x_1+x_2+x_3=1$. Transferring mass $\varepsilon>0$ from a smaller
coordinate to a larger unsaturated one increases $r^2$ by
$2\varepsilon(x_i-x_j)+2\varepsilon^2>0$. Since
$2\alpha<1<3\alpha$, the maximum occurs at
$
(x_1,x_2,x_3)
=
(\alpha,\alpha,\alpha^2),
$
and therefore
$
r_{\max}^2
=
2\alpha^2+\alpha^4
=
23-16\sqrt2.
$

For Rodrigues coordinates,
$r=\tan(\theta/2)$, so
\[
\cos\theta
=
\frac{1-r^2}{1+r^2}
\ge
\frac{1-(23-16\sqrt2)}{1+(23-16\sqrt2)}
=
\frac{2\sqrt2-1}{4}
=
c_\ast.
\]
For a suitable $C\in\mathcal C_3$, the relative rotation $CO^T$
therefore has rotation angle $\theta$ satisfying this bound.
Its eigenvalues are $1,e^{i\theta},e^{-i\theta}$, and hence
\[
\lambda_{\min}
\left[
\operatorname{Re}(CO^T)
\right]
=
\cos\theta
\ge
c_\ast.
\]
\end{proof}

The lemma also gives $S(O)\ge S_\ast$: for the chosen $C$,
\[
S(O)
\ge
\operatorname{tr}(CO^T)
=
1+2\cos\theta
\ge
S_\ast.
\]
Since rephasings preserve $S$, every unistochastic boundary
configuration satisfies $S(R)\ge S_\ast$.

\subsection{Exact extremality of the qutrit endpoint}
\label{app:exact_extremality}

We now show that the endpoint of the $TTC$ ridge attains the full
monomial optimization.

Let
$
a:=\frac1{\sqrt2},
b:=\frac12,
$
and define
\begin{equation}
R_\ast
=
\begin{pmatrix}
a & 0 & a\\
-b & a & b\\
-b & -a & b
\end{pmatrix}.
\label{eq:app_Rstar}
\end{equation}
A direct calculation gives
$R_\ast^TR_\ast=I$ and $\det R_\ast=1$.

\begin{lemma}
\label{lem:Rstar_exact}
For the matrix $R_\ast$ in Eq.~\eqref{eq:app_Rstar},
$\mu(R_\ast)=c_\ast$.
\end{lemma}

\begin{proof}
Since $R_\ast\in SO(3)$, it is a three-dimensional rotation with
$
\operatorname{tr}R_\ast
=
2a+b
=
\sqrt2+\frac12
=
1+2c_\ast.
$
Hence its rotation angle $\theta_\ast$ satisfies
$\cos\theta_\ast=c_\ast$. Taking $M=I$ in
Eq.~\eqref{eq:app_mu} gives
$
\lambda_{\min}
\left[
\operatorname{Re}(R_\ast^T)
\right]
=
c_\ast,
$
so $\mu(R_\ast)\ge c_\ast$.

For the reverse inequality, let
$M=P_\pi D$, where
$D=\operatorname{diag}(e^{i\phi_1},e^{i\phi_2},e^{i\phi_3})$, and set
$H_M:=\operatorname{Re}(MR_\ast^T)$.

For $\pi=(2,1,3)$ and $\pi=(3,1,2)$,
$\pi^{-1}(1)=2$. Since $(R_\ast)_{12}=0$,
$(H_M)_{11}=0$, so
$\lambda_{\min}(H_M)\le0<c_\ast$.

For the remaining four permutations
$
(1,2,3),\quad
(1,3,2),\quad
(2,3,1),\quad
(3,2,1),
$
choose the principal index pairs
$(2,3)$, $(1,2)$, $(2,3)$, and $(1,3)$, respectively.
In each case the corresponding principal submatrix has the form
\[
B
=
\begin{pmatrix}
p & z\\
\overline z & q
\end{pmatrix},
\qquad
p\le a,\quad q\le b,
\]
with
$
z
=
\frac12
\left(
\varepsilon_1 a e^{-i\phi}
+
\varepsilon_2 b e^{i\psi}
\right)$,
$
|\varepsilon_1|=|\varepsilon_2|=1.
$
Thus
\[
|z|\ge\frac{a-b}{2}.
\]

By Cauchy's interlacing theorem,
$\lambda_{\min}(H_M)\le\lambda_{\min}(B)$. Since
\[
B
\preceq
\begin{pmatrix}
a & z\\
\overline z & b
\end{pmatrix},
\]
we obtain
\[
\lambda_{\min}(H_M)
\le
\frac{
a+b-\sqrt{(a-b)^2+4|z|^2}
}{2}
\le
\frac{
a+b-\sqrt2\,(a-b)
}{2}
=
c_\ast.
\]
Hence $\mu(R_\ast)\le c_\ast$, and therefore
$\mu(R_\ast)=c_\ast$.
\end{proof}

Finally, after the column permutation $(1,3,2)$,
\[
|R_\ast|
\sim
\begin{pmatrix}
\dfrac1{\sqrt2} & \dfrac1{\sqrt2} & 0\\[2mm]
\dfrac12 & \dfrac12 & \dfrac1{\sqrt2}\\[2mm]
\dfrac12 & \dfrac12 & \dfrac1{\sqrt2}
\end{pmatrix}.
\]
Thus $R_\ast$ is exactly the endpoint of the $TTC$ ridge with
$
a=\frac1{\sqrt2},
\qquad
b=\frac12,
\qquad
t=0,
\qquad
S=S_\ast.
$

\subsection{Completion of the proof}
\label{app:qutrit_completion}
Let $R\in U(3)$ be arbitrary. By
Lemma~\ref{lem:matching_lower_bound},
$S(R)\ge\sqrt3$. If $S(R)\ge S_\ast$,
Lemma~\ref{lem:large_matching} gives $\mu(R)\ge c_\ast$.

It remains to consider
$\sqrt3\le S(R)<S_\ast$. In the unistochastic interior, the fixed-$S$
analysis above shows that
$
\Phi(R)\le\Phi_{\rm TTC}(S).
$
Since $\Phi=\Psi(S,\kappa)$ is strictly decreasing in $\kappa$,
$
\kappa(R)\ge\kappa_{\rm TTC}(S).
$
Together with $\mu(R)\ge\kappa(R)$ and
Lemma~\ref{lem:TTC_bound}, this gives
$\mu(R)>c_\ast$ throughout the interior hard region.
On the unistochastic boundary,
Lemma~\ref{lem:orthostochastic} gives the same universal bound.
Therefore
\begin{equation}
\mu(R)
\ge
\frac{2\sqrt2-1}{4}
=
c_\ast,
\qquad
R\in U(3).
\label{eq:app_global_lower}
\end{equation}
For every nondegenerate qutrit state,
$
D_{\mathcal I}^{U}(\rho)
\le
\arccos\left(
\frac{2\sqrt2-1}{4}
\right).
$

Finally, Lemma~\ref{lem:Rstar_exact} gives
$\mu(R_\ast)=c_\ast$. Hence
$
\inf_{R\in U(3)}\mu(R)
=
\frac{2\sqrt2-1}{4}.
$
Since
$D_{\mathcal I}^{U}(\rho)=\arccos\mu(R)$
for nondegenerate qutrits and $\arccos$ is decreasing on $[-1,1]$,
$
\sup_{\rho\ {\rm nondegenerate}}
D_{\mathcal I}^{U}(\rho)
=
\arccos\left(
\frac{2\sqrt2-1}{4}
\right).
$
This completes the proof of
Theorem~\ref{thm:qutrit_tight_bound}.
\end{document}